\documentclass[12pt, draftclsnofoot, journal, letter, onecolumn]{IEEEtran}

\usepackage{graphicx}
\usepackage{epsfig}
\usepackage{latexsym}
\usepackage{amsfonts}
\usepackage{here}
\usepackage{rawfonts}
\usepackage[latin1]{inputenc}
\usepackage[T1]{fontenc}
\usepackage{calc}
\usepackage{capitalgreekitalic}
\usepackage{url}
\usepackage{enumerate}
\usepackage{color}
\usepackage[tbtags]{amsmath}
\usepackage{amssymb}
\usepackage{upref}
\usepackage{epic,eepic}
\usepackage{times}
\usepackage{dsfont}
\usepackage{comment}
\usepackage{cite}

\newtheorem{theorem}{{\bf Theorem}}

\newcommand{\qed}{\nobreak \ifvmode \relax \else
  \ifdim\lastskip<1.5em \hskip-\lastskip
  \hskip1.5em plus0em minus0.5em \fi \nobreak
  \vrule height0.75em width0.5em depth0.25em\fi}

\newcounter{step}
\newlength{\totlinewidth}
  {\end{list}%
  \rule{\linewidth}{1pt}}
\newcounter{substep}

  {\end{list}}

\newlength{\aligntop}
\newlength{\alignbot}
 

\IEEEoverridecommandlockouts

\begin{document}

\title{Integer Forcing-and-Forward Transceiver Design for MIMO Multi-Pair Two-Way Relaying}
\author{\authorblockN{Seyed Mohammad Azimi-Abarghouyi, Masoumeh Nasiri-Kenari, \emph{Senior Member, IEEE}, and Behrouz Maham, \emph{Member, IEEE}}\\
    \thanks{
    Seyed Mohammad Azimi-Abarghouyi and Masoumeh Nasiri-Kenari are with Electrical Engineering Department, Sharif University of Technology, Tehran, Iran. Behrouz Maham is with School of ECE, College of Engineering, University of Tehran, Iran.
Emails:
\protect\url{sm_azimi@ee.sharif.edu, mnasiri@sharif.edu, bmaham@ut.ac.ir}. } }

\maketitle

\begin{abstract}
In this paper, we propose a new transmission scheme, named as Integer Forcing-and-Forward (IFF), for communications among multi-pair multiple-antenna users in which each pair exchanges their messages with the help of a single multi antennas relay in the multiple-access and broadcast phases. The proposed scheme utilizes Integer Forcing Linear Receiver (IFLR) at relay, which uses equations, i.e., linear integer-combinations of messages, to harness the intra-pair interference. Accordingly, we propose the design of mean squared error (MSE) based transceiver, including precoder and projection matrices for the relay and users, assuming that the perfect channel state information (CSI) is available. In this regards, in the multiple-access phase, we introduce two new MSE criteria for the related precoding and filter designs, i.e., the sum of the equations MSE (Sum-Equation MSE) and the maximum of the equations MSE (Max-Equation MSE), to exploit the equations in the relay. In addition, the convergence of the proposed criteria is proven as well. Moreover, in the broadcast phase, we use the two traditional MSE criteria, i.e. the sum of the users' mean squred errors (Sum MSE) and the maximum of the users' mean squared errors (Max MSE), to design the related precoding and filters for recovering relay's equations by the users. Then, we consider a more practical scenario with imperfect CSI. For this case, IFLR receiver is modified, and another transceiver design is proposed, which take into account the effect of channels estimation error. We evaluate the performance of our proposed strategy and compare the results with the conventional amplify-and-forward (AF) and denoise-and-forward (DF) strategies for the same scenario. The results indicate the substantial superiority of the proposed strategy in terms of the outage probability and the sum rate.
\end{abstract}

\section{Introduction}
Two way relaying communications have recently attracted considerable attentions due to their various applications. In this communication scenario, two users attempt to communicate with each other with the help of a relay. To this end, physical layer network coding (PLNC) [1] along with the conventional decode-and-forward (DF) or amplify-and-forward (AF) relaying strategies has been commonly employed [2-4] to improve the system throughput [5]. 

A novel relaying technique, known as compute-and-forward (CMF) [6], has been designed for multi users applications with the aim of increasing the physical layer network coding throughput. In this scheme, each relay, based on a received noisy combination of simultaneously transmitted signals of the users, attempts to recover an equation, i.e., a linear integer-combination, of users' messages, instead of recovering each individual message separately. To enable the relay to recover the equation, the CMF scheme is usually implemented based on using a proper lattice code [7]. Since the equation coefficients are selected according to the channel coefficients, this method is also called physical layer network coding [8]. The relay then transmits the decoded equation to the destination. The destination recovers the desired messages by receiving enough number of decoded equations from the relays. In fact, in contrast to conventional AF and DF relaying techniques, the CMF method exploits rather than combats the interference towards a better network performance. By applying CMF in point-to-point MIMO systems, a linear receiver, named as integer forcing linear receiver (IFLR) has been proposed in [9] in which sufficient independent equations with maximum rate are recovered to extract the users' messages.
\subsection{Motivation and Related Work}
Since the number of wireless communication users will continuously increase, independently designed one-pair two-way relay systems can scarcely accommodate a vast number of users. That is, with simultaneously transmission of pairs of users, the messages interfere with each other, and hence, arbitrary transmission and reception of the messages are not an efficient solution. To solve the problem, in [10-12], centralized designed MIMO multi-pair two-way transmission schemes with the help of a multi antenna relay have been proposed. In [10-11], the AF method has been utilized in the relay. That is, the relay simply amplifies and forwards the received signal. In [12], the DF relaying is used, as a scheme named Denoise-and-Forward, in which the relay after applying projection filter, first decodes each pair signal aligned messages individually and then precodes and transmits the decoded messages. The design criteria for precoder and projection filters in [10-12] is the minimization of the sum of the users' mean squared errors (Sum MSE). In [10], the maximization of the users' mean squared errors (Max MSE) is also considered for the transeiver design. In the simple case of single antenna one-pair two-way relay system, we have applied CMF by introducing the aligned compute-and-forward (A-CMF) scheme [13] which outperforms AF and DF based schemes significantly. 

\subsection{Contributions and Outcomes}
In this paper, we consider a more general case of two-way communications that involves multiple pairs of multiple-antenna source nodes with considering both multiple-access and broadcast phases. 

\subsubsection{Integer Forcing-and-Forward} We propose a new transmission scheme named Integer Forcing-and-Forward (IFF). We exploit the signal alignment proposed in [14-15] such that the two signals received from two users in a pair can be network coded together in the relay. Furthermore, we apply IFLR to harness the inter-pair interference in terms of equations. In the proposed scheme, the equations are decoded with higher rate than individual messages in the relay. In addition, after transmitting all recovered equations to the users, different ways to select the equations that each user needs to recover its pair's message can be utilized. Therefore, our scheme has two superiorities in comparison with the DF based scheme in [12], in which each pair message is recovered for the transmission to the respective user. 

\subsubsection{Equation Based MSE Criteria and Transceiver Design} In the proposed scheme, the precoder at transmitting nodes, including the users and the relay, and the projection filter at receiving nodes are designed based on minimizing the MSE criteria. \textit{For the first time}, we introduce the sum of the equations' mean squared errors (Sum-Equation MSE) and the maximum of the equations' mean squared errors (Max-Equation MSE) criteria for the equation recovery problem associated with the multiple-access phase precoding and filter design. These proposed Equation based MSE algorithms are proven to be convergent. Moreover, we use traditional MSE criteria, i.e. Sum MSE and Max MSE, proposed for the individual message recovery, for the broadcast phase precoding and filter design. By means of alternating optimization approach, we present tractable solutions for these MSE problems. We evaluate the performance of our proposed scheme and compare the results with those of the previous methods. Our numerical results indicate that the proposed scheme substantially outperforms the previous methods in terms of the outage probability and the network throughput. In addition, the Max based MSE precoding design, using Max-Equation MSE in multiple-access phase and Max MSE in broadcast phase, shows a better performance than Sum based MSE precoding design, using Sum-Equation MSE and Sum MSE, at the expense of more complexity. 

\subsubsection{Integer Forcing-and-Forward with Channel Estimation Error} We extend our proposed schemes for the case of imperfect channel state information (Imperfect CSI). At first, we propose Modified IFLR, taking to account the effect of channel estimation errors in the conventional IFLR receiver structure. Then, accordingly, a robust transceiver design is proposed. Simulation results show that the robust design improves the performance of the non-robust design, based on assuming the exact knowledge of CSI at the presence of the channel estimation error. 


The remainder of this paper is organized as follows. In Section II, the system model and the Integer Forcing-and-Forward scheme are briefly described. Section III presents the transceiver precoder and projection filters design by assuming that a perfect knowledge of CSI is available. In Section IV, the modified IFLR and related design are presented. Numerical results are given in Section V. Finally, Section VI concludes the paper.

\textbf{Notations:} The superscripts $\mathbf{v}^*$ and $\left| {\left| \mathbf{v} \right|} \right|$ stand for conjugate transposition and norm of vector $\mathbf{v}$, respectively. $\text{Tr}\left( \mathbf{A} \right)$, ${\mathbf{A}^\dag }$, and $(\mathbf{A})_i$ stand for trace, pseudo inverse, and the $i$-th column vector of matrix $\mathbf{A}$. The symbol $\left| x \right|$ is the absolute value of the scalar $x$, while ${\log ^ + }\left( x \right)$ denotes ${\rm{max}}\left\{ {\log \left( x \right),0} \right\}$. $\mathbb{E}\{\cdot\}$ is the expectation of a random variable $x$. $\mathbf{I}$ denotes identity matrix. $\text{vec}(.)$ and $\text{mat}(.)$ represent the matrix vectorization and its inverse operation, respectively. $ \otimes$ denotes the Kronecker product.

\section{System Model and Integer Forcing-and-Forward (IFF) Scheme}
We consider a MIMO multi-pair two-way relaying system with $K$ pairs, i.e., $2K$ users, and one relay R, as shown in Fig. 1. In this system, in each pair $k$, users $k$ and $\bar k \buildrel \Delta \over = \text{mod}\left( {K + k} \right)$ attempt to exchange their messages, i.e.  messages vectors $\mathbf{w}_k$ and ${\mathbf{w}_{\bar k}}$ each with dimension of $L_{k}$ by the help of the relay R. Each user $k$ exploits a lattice encoder with normalized power to project its message vector $\mathbf{w}_k$ to a length-$n$ complex-valued codeword vector  $\mathbf{s}_k$ such that ${\left| {\left| {{\mathbf{s}_k}} \right|} \right|^2} \le n$. We assume that user $k$ and relay R have $N_k$ and $N_r$ antennas, respectively. The matrix $\mathbf{H}_k$ denotes the channel matrix from user $k$ to the relay, with dimension ${N_r} \times {N_k}$. The elements of $\mathbf{H}_k$ are assumed to be independent identically distributed, i.i.d, Rayleigh variables with variance $\sigma _k^2$. User $k$ precodes its message $\mathbf{s}_k$ with matrix $\mathbf{V}_k$, with dimension ${N_k} \times {L_k}$, and transmits the precoded signal $\mathbf{x}_k=\mathbf{V}_k \mathbf{s}_k$. For each pair $k$, the following power constraint on the sum power is considered:
\begin{eqnarray}
\text{Tr}\left( {{\mathbf{V}_k}\mathbf{V}_k^*} \right) + \text{Tr}\left( {{\mathbf{V}_{\bar k}}\mathbf{V}_{\bar k}^*} \right) \le {P_k}, k = 1, \ldots ,K.
\end{eqnarray}
\begin{figure}[tb!]
\centering

\includegraphics[width =5in]{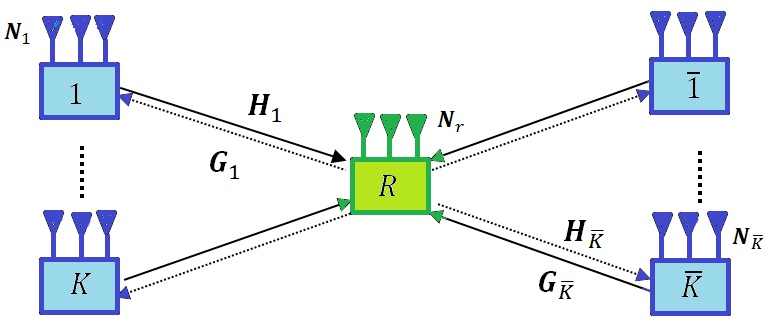}

\caption{MIMO Multi-pair Two-way Relay System}

\end{figure}
In the Integer Forcing-and-Forward (IFF) scheme, we use Multiple Access Broadcast (MABC) protocol introduced in [16]. That is, in the first time slot, named multiple access phase, the users transmit simultaneously, and therefore, the received signal by the relay R can be written as
\begin{eqnarray}
{\mathbf{y}_r} = \mathop \sum \limits_{k = 1}^{2K} {\mathbf{H}_k}{\mathbf{V}_k}{\mathbf{s}_k} + {\mathbf{z}_r} = \mathop \sum \limits_{k = 1}^K {\mathbf{H}_k}{\mathbf{V}_k}{\mathbf{s}_k} + \mathop \sum \limits_{ k = 1}^{ K} {\mathbf{H}_{\bar k}}{\mathbf{V}_{\bar k}}{\mathbf{s}_{\bar k}} + {\mathbf{z}_r},
\end{eqnarray}    
where $\mathbf{z}_r$ denotes the received noise at the relay and has Gaussian distribution with variance $\sigma _r^2$.

We use the signal alignment scheme proposed in [14-15] such that the received signals from the users in each pair $k$ to be aligned in the relay, i.e., 
\begin{eqnarray}
{\mathbf{H}_k}{\mathbf{V}_k} = {\mathbf{H}_{\bar k}}{\mathbf{V}_{\bar k}}, k = 1, \ldots ,K.
\end{eqnarray}
Hence, the user $\bar k$ precoder, ${\mathbf{V}_{\bar k}}$, versus its pair precoder, i.e., $\mathbf{V}_k$, is given by [17]
\begin{eqnarray}
{\mathbf{V}_{\bar k}} = \mathbf{H}_{\bar k}^\dag {\mathbf{H}_k}{\mathbf{V}_k},
\end{eqnarray}
where $\mathbf{H}_{\bar k}^\dag$ is the pseudo inverse of ${\mathbf{H}_{\bar k}}$. We can rewrite $\mathbf{y}_r$ as
\begin{eqnarray}
{\mathbf{y}_r} = \mathop \sum \limits_{k = 1}^K {\mathbf{H}_k}{\mathbf{V}_k}{\mathbf{\dot s}_k} + {\mathbf{z}_r},
\end{eqnarray}
where we define ${\mathbf{\dot s}_k} = {\mathbf{s}_k} + {\mathbf{s}_{\bar k}}$, named as the $k$-th pair sum message. In addition, we can rewrite $\mathbf{y}_r$ in a different form, similar to MIMO point-to-point channel, as
\begin{eqnarray}
{\mathbf{y}_r} = \mathbf{H}\mathbf{\dot S} + {\mathbf{z}_r},
\end{eqnarray}                                
where $\mathbf{\dot S} \buildrel \Delta \over = \left[ {{{\mathbf{\dot s}}_1}, \ldots ,{{\mathbf{\dot s}}_K}} \right]^*$ and
\begin{eqnarray}
\mathbf{H} \buildrel \Delta \over = \left[ {{\mathbf{H}_1}{\mathbf{V}_1}, \ldots ,{\mathbf{H}_K}{\mathbf{V}_K}} \right].
\end{eqnarray}
Moreover, the constrain in (1) can be rewritten as
\begin{eqnarray}
\text{Tr}\left( {{\mathbf{V}_k}\mathbf{V}_k^*} \right) + \text{Tr}\left( {\mathbf{H}_{\bar k}^\dag {\mathbf{H}_k}{\mathbf{V}_k}\mathbf{V}_k^*\mathbf{H}_k^*\mathbf{H}{{_{\bar k}^\dag }^*}} \right) \le {P_k}, k=1,...,K.
\end{eqnarray}
As seen in (5) and (8), the signal alignment makes it possible to consider only one user's parameters from each pair, and therefore, the MSE criterion, as will be discussed in Section III, can be more easily tractable. Note that each user in a pair can recover the other user's message, by having the related sum pair message.

In the second time slot, named broadcast phase, considering the CMF concept, we transmit equations of the users' messages rather than individual messages in the DF based scheme [12]. At first, the relay recovers $L$, $L \buildrel \Delta \over = {L_1} +  \ldots  + {L_K}$, independent equations from the received signal $\mathbf{y}_r$ by applying IFLR method [9], which is developed for MIMO channels. The $L$ independent equations, with coefficient vectors (ECVs) ${\mathbf{a}_k}, k=1,...,L$, (totally shown by matrix $\mathbf{A}$), can be solved to recover the $L$ pair sum messages, i.e. all of the pair sum messages. 

The answer of the equation with ECV $\mathbf{a}_k$, i.e. $t_k$, can be recovered by quantization of the projected received signal $\mathbf{y}_k$:
\begin{eqnarray}
{t_k} =\mathbf{a}_k^*\mathbf{\dot S}= Q\left( {\mathbf{b}_k^*{\mathbf{y}_k}} \right),
\end{eqnarray}
where $Q(.)$ denotes lattice quantizer function, and $\mathbf{b}_k$, with dimension ${N_r} \times {1}$, is the projection vector. The vector $\mathbf{b}_k$ is given by [9]
\begin{eqnarray}
\mathbf{b}_k^* = \mathbf{a}_k^*{\mathbf{H}^*}{\left( {\frac{{\sigma _r^2}}{2}\mathbf{I} + \mathbf{H}{\mathbf{H}^*}} \right)^{ - 1}}\mathbf{H}.
\end{eqnarray}
and the computation rate of this equation, i.e. the detecting rate, is given by [9]
\begin{eqnarray}
{R_k} = \log _2^ + \left( {\frac{1}{{\sigma _r^2/2 ||\mathbf{b}_k||^2+ {{\left| {\left| {{\mathbf{b}_k}\mathbf{H} - {\mathbf{a}_k}} \right|} \right|}^2}}}} \right).
\end{eqnarray}

Then, the relay puts the recovered equations in vector $\mathbf{t} = {\left[ {{t_1}, \ldots ,{t_L}} \right]^*}$. After projecting $\mathbf{t}$ with matrix $\mathbf{W}$, with diminsion ${N_r} \times L$, the relay transmits the result to the users. We consider power constraint $P_r$ for the relay transmission, i.e.,
\begin{eqnarray}
\text{Tr}\left( {\mathbf{W}{\mathbf{W}^*}} \right) \le {P_r}.
\end{eqnarray}
We assume that $\mathbf{G}_k$, with the dimension of ${N_k} \times {N_r}$, is the channel coefficient matrix from relay R to the user $k$. The elements of the matrix are assumed i.i.d Rayleigh variables with the identical variance $\sigma _k^2$.  
The received signal by each user $k$ is given by
\begin{eqnarray}
{\mathbf{y}_k} = {\mathbf{G}_k}\mathbf{W}\mathbf{t} + {\mathbf{z}_k},
\end{eqnarray}
where $\mathbf{z}_k$ denotes the receiver noise, having Gaussian distribution with variance $\sigma _u^2$.
User $k$ exploits projection filter $\mathbf{D}_k$, a matrix with dimension $L \times {N_k}$, to recover equation vector $\mathbf{t}$ using a traditional linear receiver as
\begin{eqnarray}
\mathbf{t} = Q\left( {{\mathbf{D}_k}{\mathbf{y}_k}} \right).
\end{eqnarray}
According to (13), which shows a point-to-point MIMO channel, the rate of recovering the equation $t_i$ by user $k$ is given by [20]
\begin{eqnarray}
{\tilde R_i^k} = {\rm{lo}}{{\rm{g}}_2}\left( {1 + \frac{{{{\left| {\left| {{{\left( {\mathbf{D}_k^*} \right)}_i}{{\left( {{\mathbf{G}_k}\mathbf{W}} \right)}_i}} \right|} \right|}^2}}}{{{{\left| {\left| {{\left( {\mathbf{D}_k^*} \right)}_i} \right|} \right|}^2} + \mathop \sum \nolimits_{l \ne i} {{\left| {\left| {{{\left( {\mathbf{D}_k^*} \right)}_i}{{\left( {{\mathbf{G}_k}\mathbf{W}} \right)}_l}} \right|} \right|}^2}}}} \right).
\end{eqnarray}
This achievable rate can be improved using successive interference cancellation (SIC) [21]. Therefore, the overall rate of recovering the equation with ECV $\mathbf{a}_i$, i.e. $t_i$, by user $k$ is
\begin{eqnarray}
R_k^i = {\rm{min}}\left( {{R_i},{{\tilde R}_i^k}} \right),
\end{eqnarray}
where $R_i$ is given in (11). Each user, among received equations $\mathbf{t}$, uses the best ones with the maximum overall rate that can help the user to recover its pair's messages. In comparison with the DF based scheme in [12], not only higher rate is achieved by decoding equations at the relay instead of the messages [9], but also more flexibility is provided for the users, having different ways to recover their pairs' messages according to the ECVs of the transmitted relay's equations. Please note, even in the worst case, each user can still recover its pair's messages because the relay transmits independent equations in the number of all pairs' messages.

\section{MSE based Precoding and Projection Filter Design for the Perfect Channel Knowledge}
In this section, based on the poposed IFF scheme, we investigate the transceiver design, i.e., finding precoding and projection Filter matrices for all the nodes to minimize the MSE by assuming that the perfect CSI is available. According to the proposed scheme presented in Section II, we have to select the design matrices for two phases of multiple access and broadcast, separately. First, we consider the multiple-access phase, in which the users' transmitting precoding matrix $\mathbf{V}_k$ and the relay's receiving projection matrix $\mathbf{B}$ are optimized. Similarly, in Subsection III.B, we consider the broadcast phase, and obtain the relay's precoder matrix $\mathbf{W}$ and the users' projection matrix $\mathbf{D}_k$.

\subsection{Multiple-Access Phase MSE based Precoding and Projection Filter Design}
At first, we design the related matrices in the multiple-access phase by introducing Max-Equation MSE criterion for our equation based problem, to ensure QoS equivalency between different recovered equations. However, since maybe some users do not use all of the equations to recover their pairs' messages, we introduce Sum-Equation MSE criterion, which also has less complexity at the ECV search problem, as will be discussed.

\subsubsection{Max-Equation MSE based Precoding and Projection Filter Design}
From (6) and (9), the effective noise ${\varepsilon _k}$ in recovering the equation ${t_k} = \mathbf{a}_k^*\mathbf{\dot S}$ from the projection of the received signal $\mathbf{y}_r$ onto vector $\mathbf{b}_k$ is equal to
\begin{eqnarray}
{\varepsilon _k} = \mathbb{E}\left\{{\left| {\left| {\mathbf{b}_k^*{\mathbf{y}_r} - \mathbf{a}_k^*\mathbf{\dot S}} \right|} \right|^2}\right\}. 
\end{eqnarray}
Now, the users' precoding vectors $\mathbf{V}_k, k = 1, \ldots ,K$, equation matrix $\mathbf{A}$, and projection matrix $\mathbf{B}$, including the vectors $\mathbf{b}_k$ in (10), $k = 1, \ldots ,L$, must be selected so as to minimize the maximum effective noise of all of the $L$ recovering equations, i.e.,
\begin{eqnarray}
\mathop {\min }\limits_{{\mathbf{V}_k}, \mathbf{A}, \mathbf{B}} \mathop {\max }\limits_{i = 1, \ldots ,L} {\varepsilon _i}\nonumber
\end{eqnarray}
subject to
\begin{eqnarray}
\text{Tr}\left( {{\mathbf{V}_k}\mathbf{V}_k^*} \right) + \text{Tr}\left( {\mathbf{H}_{\bar k}^\dag {\mathbf{H}_k}{\mathbf{V}_k}\mathbf{V}_k^*\mathbf{H}_k^*\mathbf{H}{{_{\bar k}^\dag }^*}} \right) \le {P_k}, k=1,...,K.
\end{eqnarray}
where from (5) and (17), ${\varepsilon _k}$ can be expanded as
\begin{eqnarray}
{{\rm{\varepsilon }}_k} &=& \text{Tr}\left\{\left(\mathbf{b}_k^*\mathop \sum \limits_{i = 1}^K {\mathbf{H}_i}{\mathbf{V}_i}{\mathbf{\dot s}_i} + \mathbf{b}_k^*{\mathbf{z}_r} - \mathop \sum \limits_{i = 1}^K \mathbf{a}_{k,i}^*{\mathbf{\dot s}_i}\right){\left(\mathbf{b}_k^*\mathop \sum \limits_{i = 1}^K {\mathbf{H}_i}{\mathbf{V}_i}{\mathbf{\dot s}_i} + \mathbf{b}_k^*{\mathbf{z}_r} - \mathop \sum \limits_{i = 1}^K \mathbf{a}_{k,i}^*{\mathbf{\dot s}_i}\right)^*}\ \right\} \nonumber\\&=& 2\mathop \sum \limits_{i = 1}^K \left(\mathbf{b}_k^*{\mathbf{H}_i}{\mathbf{V}_i}\mathbf{V}_i^*\mathbf{H}_i^*{\mathbf{b}_k} - 2\mathbf{a}_{k,i}^*\mathbf{V}_i^*\mathbf{H}_i^*{\mathbf{b}_k} + \mathbf{a}_{k,i}^*{\mathbf{a}_{k,i}}\right)+\sigma _r^2\mathbf{b}_k^*{\mathbf{b}_k},
\end{eqnarray}
where $\mathbf{a}_{k,i}$ is the $i$-th pair coefficient of the $k$-th equation. By substituting $\mathbf{b}_k$ from (10) and with some straightforward simplifications, we can rewrite $\varepsilon _k$ as
\begin{eqnarray}
{{\rm{\varepsilon }}_k}=\mathbf{a}_k^*\mathbf{U}\mathbf{a}_k,
\end{eqnarray}
where
\begin{eqnarray}
\mathbf{U} \buildrel \Delta \over = \mathbf{I} - {\mathbf{H}^*}{\left( {\frac{{\sigma _r^2}}{2}\mathbf{I} + \mathbf{H}{\mathbf{H}^*}} \right)^{ - 1}}\mathbf{H}.
\end{eqnarray}
Using the alternative method, we solve the given optimization problem. That is, in the first step, assuming the precoding vectors are known, the matrix $\mathbf{A}_\text{opt}$ is obtained as
\begin{eqnarray}
\mathbf{A}_\text{opt}= \text{arg}\mathop {\min }\limits_{\mathbf{A} \in {\mathbb{Z}^{L \times L}}} \mathop {\max }\limits_{k = 1, \ldots ,L} \left( {\mathbf{a}_k^*\mathbf{U}{\mathbf{a}_k}} \right),
\end{eqnarray}
 subject to
\begin{eqnarray}
\left\{ {\begin{array}{*{20}{c}}{\mathbf{A} = \left[ {\begin{array}{*{20}{c}}{\mathbf{a}_1^*}\\ \vdots \\{\mathbf{a}_L^*}\end{array}} \right].}\\{\det \left( \mathbf{A} \right) \ne 0.}\\{{\mathbf{a}_k} \in {\mathbb{Z}^L}, k=1, \ldots ,L.}\end{array}} \right. \nonumber 
\end{eqnarray}
This optimization problem, named ECV search, can be solved efficiently by using the proposed schemes in [18-19].

In the second step, by substituting the values obtained at the first step for matrices $\mathbf{A}$ and $\mathbf{B}$, the precoding vectors are calculated as follows:

By introducing a new variable $x$ that serves an upper bound on $\varepsilon _i, \forall i$, the optimization problem of precoding matrices $\mathbf{V}_k, \forall k$ can be rewritten as 
\begin{eqnarray}
\mathop {\min }\limits_{\{ {\mathbf{V}_k};x\} } x, \nonumber
\end{eqnarray}
subject to
\begin{eqnarray}
\left\{ {\begin{array}{*{20}{c}}{{\varepsilon _i} \le x}\\{\text{Tr}\left( {{\mathbf{V}_k}\mathbf{V}_k^*} \right) + \text{Tr}\left( {\mathbf{H}_{\bar k}^\dag {\mathbf{H}_k}{\mathbf{V}_k}\mathbf{V}_k^*\mathbf{H}_k^*\mathbf{H}{{_{\bar k}^\dag }^*}} \right) \le {P_k}, k=1,...,K.}\end{array}} \right.
\end{eqnarray}
With definition of $\mathbb{H}$ and $\mathbb{V}$ as
\begin{eqnarray}
\mathbb{H} \buildrel \Delta \over = \left[ {\begin{array}{*{20}{c}}{{\mathbf{H}_1}}&{\begin{array}{*{20}{c}} \ldots &{{\mathbf{H}_K}}\end{array}}\end{array}} \right],
\end{eqnarray}
\begin{eqnarray}
\mathbb{V} \buildrel \Delta \over = \left( {\begin{array}{*{20}{c}}{\begin{array}{*{20}{c}}{{\mathbf{V}_1}}&0\\0&{{\mathbf{V}_2}}\end{array}}&{\begin{array}{*{20}{c}} \cdots &0\\ \cdots & \vdots \end{array}}\\{\begin{array}{*{20}{c}} \vdots & \vdots \\0&0\end{array}}&{\begin{array}{*{20}{c}} \ddots &0\\ \cdots &{{\mathbf{V}_K}}\end{array}}\end{array}} \right),
\end{eqnarray}
$\varepsilon _i$ in (19) can be rewritten as
\begin{eqnarray}
\varepsilon _i&=&\mathop \sum \limits_{i = 1}^K 2{ || {\mathbf{V}_i^*\mathbf{H}_i^*{\mathbf{b}_k} - {\mathbf{a}_{k,i}}} ||^2} + \sigma _r^2{\left| {\left| {{\mathbf{b}_k}} \right|} \right|^2}\nonumber\\&=&\left| {\left| {{\mathbb{V}^*}{\mathbb{H}^*}{\mathbf{b}_k} - {\mathbf{a}_k}} \right|} \right|^2 + \sigma _r^2{\left| {\left| {{\mathbf{b}_k}} \right|} \right|^2}.
\end{eqnarray}
Hence, with the help of equation $\text{vec}(\mathbf{XYZ}) = ({\mathbf{Z}^*} \otimes \mathbf{X})\text{vec}(\mathbf{Y})$ in [17], this leads to
\begin{eqnarray}
\varepsilon _i=\left|\left|\begin{array}{*{20}{c}}{{\sigma _r}{||\mathbf{b}_k||}}\\{(\mathbf{b}_k^*\mathbb{H} \otimes \mathbf{I})\text{vec}\left( {{\mathbb{V}^*}} \right) - {\mathbf{a}_k}}\end{array}\right|\right|^2.
\end{eqnarray}
Accordingly, the optimization problem of transmit precoding matrices can be rewritten as
\begin{eqnarray}
\mathop {\min }\limits_{\{ {\mathbf{V}_k};x\} } x, \nonumber
\end{eqnarray}
subject to
\begin{eqnarray}
\left\{ {\begin{array}{*{20}{c}}{\left|\left|\begin{array}{*{20}{c}}{{\sigma _r}||\mathbf{b}_k||}\\{(\mathbf{b}_k^*\mathbb{H} \otimes \mathbf{I})\text{vec}\left( {{\mathbb{V}^*}} \right) - {\mathbf{a}_k}}\end{array}\right|\right|^2 \le x}\\{{\left| {\left| {{\rm{\text{vec}}}\left( {{\mathbf{V}_k}} \right)} \right|} \right|^2} + {\left| {\left| {\text{vec}\left( {{\rm{ \mathbf{H}_{\bar k}^\dag {\mathbf{H}_k}{\mathbf{V}_k}}}} \right)} \right|} \right|^2} \le {P_k}, k=1,..., K.}\end{array}} \right.
\end{eqnarray}
This optimization problem is a Second Order Cone Programming (SOCP) problem [22] due to the fact that the objective function is linear and the constraints are second order cones. It can be efficiently solved by standard SOCP solver [23] or CVX, a software package that is developed for convex optimization problems. Algorithm 1 summarizes the above procedures.
\begin{theorem}
The proposed Max-Equation MSE minimization algorithm is convergent.
\end{theorem}
\begin{proof}
Let $\varepsilon  = \mathop {\max }\limits_{i = 1, \ldots ,L} {\varepsilon _i}$, the overall MSE. Thus, in the first step of Algorithm 1 for the iteration $j+1$, we have $\varepsilon({\mathbf{A}^{\left( {j + 1} \right)}},{\mathbf{B}^{\left( {j + 1} \right)}},{\mathbf{V}^{\left( {j} \right)}}) \le \varepsilon({\mathbf{A}^{\left( {j} \right)}},{\mathbf{B}^{\left( {j} \right)}},{\mathbf{V}^{\left( {j} \right)}})$, and in the second step, $\varepsilon({\mathbf{A}^{\left( {j + 1} \right)}},{\mathbf{B}^{\left( {j + 1} \right)}},{\mathbf{V}^{\left( {j+1} \right)}})\le \varepsilon({\mathbf{A}^{\left( {j+1} \right)}},{\mathbf{B}^{\left( {j+1} \right)}},{\mathbf{V}^{\left( {j} \right)}})$. Hence, $\varepsilon({\mathbf{A}^{\left( {j + 1} \right)}},{\mathbf{B}^{\left( {j + 1} \right)}},{\mathbf{V}^{\left( {j+1} \right)}}) \le \varepsilon({\mathbf{A}^{\left( {j} \right)}},{\mathbf{B}^{\left( {j} \right)}},{\mathbf{V}^{\left( {j} \right)}})$ at the end of iteration $j+1$. Therefore, in each iteration, the overall MSE decreases, which is lower bounded by zero. Hence, the proposed Max-MSE minimization algorithm is convergent.
\end{proof}
\begin{table}[t]
  \centering
  \caption{
    Algorithm 1: Multiple-Access Phase Max-Equation based MSE Precoding and Projection Filter Design}
    \vspace*{-1em}
    \begin{tabular}{l}
      \hline
     Initialize $\mathbf{V}_k^{\left( 0 \right)}, \forall k$ and $\delta $
\\
Iterate
\\
\hspace*{+8pt}1.ECV search: update ${\mathbf{A}^{\left( {j + 1} \right)}}$ and ${\mathbf{B}^{\left( {j + 1} \right)}}$ from (22) and (10) for fixed $\mathbf{V}_k^{\left( j \right)}, \forall k$
\\
\hspace*{+8pt}2.Update $\mathbb{V}^{(j+1)}$, i.e. $\mathbf{V}_k^{\left( j+1 \right)}, \forall k$, by solving the SOCP problem of (28) for fixed ${\mathbf{A}^{\left( {j + 1} \right)}}$ and ${\mathbf{B}^{\left( {j + 1} \right)}}$ 
\\
Until $\text{Tr}\left( {\left( {\mathbf{V}_k^{\left( {j + 1} \right)} - \mathbf{V}_k^{\left( j \right)}} \right){{\left( {\mathbf{V}_k^{\left( {j + 1} \right)} - \mathbf{V}_k^{\left( j \right)}} \right)}^*}} \right) \le \delta, \forall k$ 

    \vspace*{.5em}
\\

      \hline
    \end{tabular}

\end{table}
\subsubsection{Sum-Equation MSE based Precoding and Projection Filter Design}
The optimization problem, which minimizes the total effective noise from all of the $L$ recovering equations, can be considered as
\begin{eqnarray}
\mathop {\min }\limits_{{\mathbf{V}_k},\mathbf{A},\mathbf{B}} \ \varepsilon=\mathop \sum \limits_{i = 1}^L {\varepsilon _i} = \mathop \sum \limits_{i = 1}^L \mathbb{E}\left\{ {\left| {\left| {\mathbf{b}_i^*{\mathbf{y}_r} - \mathbf{a}_i^*\mathbf{\dot S}} \right|} \right|^2}\right\},\nonumber
\end{eqnarray}
subject to
\begin{eqnarray}
\text{Tr}\left( {{\mathbf{V}_k}\mathbf{V}_k^*} \right) + \text{Tr}\left( {\mathbf{H}_{\bar k}^\dag {\mathbf{H}_k}{\mathbf{V}_k}\mathbf{V}_k^*\mathbf{H}_k^*\mathbf{H}{{_{\bar k}^\dag }^*}} \right) \le {P_k}, k=1,...,K.
\end{eqnarray}
where
\begin{eqnarray}
\varepsilon = \mathop \sum \limits_{k = 1}^L \mathop \sum \limits_{i = 1}^K 2(\mathbf{b}_k^*{\mathbf{H}_i}{\mathbf{V}_i}\mathbf{V}_i^*\mathbf{H}_i^*{\mathbf{b}_k} - 2\mathbf{a}_{k,i}^*\mathbf{V}_i^*\mathbf{H}_i^*{\mathbf{b}_k} + \mathbf{a}_{k,i}^*{\mathbf{a}_{k,i}})+ \sigma _r^2\mathbf{b}_k^*{\mathbf{b}_k},
\end{eqnarray}
and according to (20), we have
\begin{eqnarray}
\varepsilon = \mathop \sum \limits_{k = 1}^L \mathbf{a}_k^*\mathbf{U}\mathbf{a}_k.
\end{eqnarray}
We can rewrite (31) in a simpler form as
\begin{eqnarray}
\varepsilon = \text{Tr}(\mathbf{A}\mathbf{U}\mathbf{A}^*).\nonumber
\end{eqnarray}
Again, we solve this problem by using the alternative method. In the first step, the matrix $\mathbf{A}_\text{opt}$ is obtained as
\begin{eqnarray}
\mathbf{A}_\text{opt}= \text{arg}\mathop {\min }\limits_{\mathbf{A} \in {\mathbb{Z}^{L \times L}}} \text{Tr}(\mathbf{A}\mathbf{U}\mathbf{A}^*)
\end{eqnarray}
 subject to
\begin{eqnarray}
\left\{ {\begin{array}{*{20}{c}}{\mathbf{A} = \left[ {\begin{array}{*{20}{c}}{\mathbf{a}_1^*}\\ \vdots \\{\mathbf{a}_L^*}\end{array}} \right].}\\{\det \left( \mathbf{A} \right) \ne 0.}\\{{\mathbf{a}_k} \in {\mathbb{Z}^L}, k=1, \ldots ,L.}\end{array}} \right. \nonumber 
\end{eqnarray}
We can solve this problem by using the proposed schemes in [18-19] with some straightforward changes. However, since we can optimize $\mathbf{A}$ at once, this problem is significantly more simple and tractable than (22).

In the second step, the precoding vectors can be calculated as follows. 
The KKT conditions for the $k$-th pair precoder $\mathbf{V}_k$ can be written as
\begin{eqnarray}
{\nabla _{{\mathbf{V}_k}}}\varepsilon  + {\mu _k}{\nabla _{{\mathbf{V}_k}}}\left\{ {\text{Tr}\left( {{\mathbf{V}_k}\mathbf{V}_k^*} \right) + \text{Tr}\left( {\mathbf{H}_{\bar k}^\dag {\mathbf{H}_k}{\mathbf{V}_k}\mathbf{V}_k^*\mathbf{H}_k^*\mathbf{H}{{_{\bar k}^\dag }^*}} \right) - {P_k}} \right\}=0,
\end{eqnarray}
\begin{eqnarray}
{\mu _k}\left( {\text{Tr}\left( {{\mathbf{V}_k}\mathbf{V}_k^*} \right) + \text{Tr}\left( {\mathbf{H}_{\bar k}^\dag {\mathbf{H}_k}{\mathbf{V}_k}\mathbf{V}_k^*\mathbf{H}_k^*\mathbf{H}{{_{\bar k}^\dag }^*}} \right) - {P_k}} \right) = 0, k = 1, \ldots ,K,
\end{eqnarray}
where ${\mu _k}$ is the KKT coefficient related to the pair $k$. From (30) and (33), we obtain
\begin{eqnarray}
2\mathop \sum \limits_{i = 1}^L (\mathbf{H}_k^*{\mathbf{b}_i}\mathbf{b}_i^*{\mathbf{H}_k}{\mathbf{V}_k} - \mathbf{H}_k^*{\mathbf{b}_i}\mathbf{a}_{k,i}^*) + {\mu _k}\left( {{\mathbf{V}_k} + \mathbf{H}_k^*\mathbf{H}{{_{\bar k}^\dag }^*}\mathbf{H}_{\bar k}^\dag {\mathbf{H}_k}{\mathbf{V}_k}} \right)=0.
\end{eqnarray}
Hence, we have 
\begin{eqnarray}
\mathbf{V}_k ={\left(\mathbf{H}_k^*\mathop \sum \limits_{i = 1}^L {\mathbf{b}_i}\mathbf{b}_i^*{\mathbf{H}_k} + \frac{1}{2}{\mu _k}\left( {\mathbf{I} + \mathbf{H}_k^*\mathbf{H}{{_{\bar k}^\dag }^*}\mathbf{H}_{\bar k}^\dag {\mathbf{H}_k}} \right)\right)^{ - 1}}\mathop \sum \limits_{i = 1}^L \mathbf{H}_k^*{\mathbf{b}_i}\mathbf{a}_{k,i}^*.
\end{eqnarray}
From (4), its pair can be calculated as 
\begin{eqnarray}
{\mathbf{V}_{\bar k}} = \mathbf{H}_{\bar k}^\dag {\mathbf{H}_k}{\left(\mathbf{H}_k^*\mathop \sum \limits_{i = 1}^L {\mathbf{b}_i}\mathbf{b}_i^*{\mathbf{H}_k} + \frac{1}{2}{\mu _k}\left( {\mathbf{I} + \mathbf{H}_k^*\mathbf{H}{{_{\bar k}^\dag }^*}\mathbf{H}_{\bar k}^\dag {\mathbf{H}_k}} \right)\right)^{ - 1}}\mathop \sum \limits_{i = 1}^L \mathbf{H}_k^*{\mathbf{b}_i}\mathbf{a}_{k,i}^*.
\end{eqnarray}
Here, ${\mu _k}$ is determined from the second KKT condition given in (34). We consider two cases, namely  ${\mu _k}=0$  and ${\mu _k}>0$ as mentioned in [22, pp. 243]. If ${\mu _k}=0$ , or in other words when the optimum solution is in the feasible region, we should have $\text{Tr}\left( {{\mathbf{V}_k}\left( 0 \right)\mathbf{V}_k^*\left( 0 \right)} \right) + \text{Tr}\left( {{\mathbf{V}_{\bar k}}\left( 0 \right)\mathbf{V}_{\bar k}^*\left( 0 \right)} \right) \le {P_k}$. On the other hand, if ${\mu _k}>0$ , or equivalently, when the optimum solution is on the constraint border, we have $\text{Tr}\left( {{\mathbf{V}_k}\left( {{\mu _k}} \right)\mathbf{V}_k^*\left( {{\mu _k}} \right)} \right) + \text{Tr}\left( {{\mathbf{V}_{\bar k}}\left( {{\mu _k}} \right)\mathbf{V}_{\bar k}^*\left( {{\mu _k}} \right)} \right) - {P_k} = 0$. For the latter case, we can find ${\mu _k}>0$ efficiently by applying the bisection optimization method [22]. The above procedures are summarized in Algorithm 2. The parameter $\delta$ used in the algorithm determines the convergence tolerance.
\begin{theorem}
The proposed Sum-Equation MSE minimization algorithm is convergent.
\end{theorem}
\begin{proof}
The proof is similar to the one given for Theorem 1.
\end{proof}
\begin{table}[t]
  \centering
  \caption{
    Algorithm 2: Multiple-Access Phase Sum-Equation MSE based Precoding and Projection Filter Design}
    \vspace*{-1em}
    \begin{tabular}{l}
      \hline
      Initialize $\mathbf{V}_k^{\left( 0 \right)}, \forall k$ and $\delta $
\\
Iterate
\\
\hspace*{+8pt}1.ECV search: update ${\mathbf{A}^{\left( {j + 1} \right)}}$ and ${\mathbf{B}^{\left( {j + 1} \right)}}$ from (32) and (10) for fixed $\mathbf{V}_k^{\left( j \right)}, \forall k$
\\
\hspace*{+8pt}2.Update $\mathbf{V}_k^{\left( j+1 \right)}, \forall k$ with finding $\mu _k^{\left( {j + 1} \right)}$ for fixed ${\mathbf{A}^{\left( {j + 1} \right)}}$ and ${\mathbf{B}^{\left( {j + 1} \right)}}$ 
\\
Until $\text{Tr}\left( {\left( {\mathbf{V}_k^{\left( {j + 1} \right)} - \mathbf{V}_k^{\left( j \right)}} \right){{\left( {\mathbf{V}_k^{\left( {j + 1} \right)} - \mathbf{V}_k^{\left( j \right)}} \right)}^*}} \right) \le \delta, \forall k$ 
\vspace*{.5em}
   \\  
      \hline
    \end{tabular}

\end{table}

\subsection{Broadcast Phase MSE based Precoding and Projection Filter Design}
In the broadcast phase, for recovering the transmitted equations $\mathbf{t}$ in each user, we use traditional Sum MSE and Max MSE to design the related matrices. At first, we consider Sum MSE criterion for this phase precoding and filter design. As Sum MSE can be unfair at recovery of the transmitted equations in different users, we use Max MSE to guarantee the QoS of each user, as well.
\subsubsection{Sum MSE based Precoding and Projection Filter Design}
From (13) and (14), the decoding noise $\tilde \varepsilon _k$ for recovering the equation vector $\mathbf{t}$ by each user $k$ is equal to:
\begin{eqnarray}
{\tilde \varepsilon _k} = \mathbb{E}\left\{ {{{\left| {\left| {{\mathbf{D}_k}{\mathbf{y}_k} - \mathbf{t}} \right|} \right|}^2}} \right\}.
\end{eqnarray}
Now, the relay's precoder matrix $\mathbf{W}$ and users' projecting vectors $\mathbf{D}_k$ are selected in order to minimize the total decoding noises due to all users as:
\begin{eqnarray}
\mathop {\min }\limits_{\mathbf{W},{\mathbf{D}_k}}\tilde \varepsilon= \mathop \sum \limits_{k = 1}^{2K} {\tilde \varepsilon _k} = \mathop \sum \limits_{k = 1}^{2K} \mathbb{E}\left\{ {{{\left| {\left| {{\mathbf{D}_k}{\mathbf{y}_k} - \mathbf{t}} \right|} \right|}^2}} \right\},\nonumber
\end{eqnarray}
subject to
\begin{eqnarray}
\text{Tr}\left( {\mathbf{W}{\mathbf{W}^*}} \right) \le {P_r}.
\end{eqnarray}
Again, this problem is solved by the alternative optimization method. In the first step, assuming the relay precoder matrix $\mathbf{W}$ is known, from (39), the users' projection vectors $\mathbf{D}_k$s are obtained as follows:
\begin{eqnarray}
\mathop {\min }\limits_{{\mathbf{D}_k}} \tilde \varepsilon ,
\end{eqnarray} 
where with some simplifications, we have                                
\begin{eqnarray}
\tilde \varepsilon  = \mathop \sum \limits_{k = 1}^{2K} \text{Tr}\left\{ {{\mathbf{D}_k}{\mathbf{G}_k}\mathbf{W}{\mathbf{W}^*}\mathbf{G}_k^*\mathbf{D}_k^* - 2{\mathbf{W}^*}\mathbf{G}_k^*\mathbf{D}_k^* + \mathbf{I} + \sigma _k^2{\mathbf{D}_k}\mathbf{D}_k^*} \right\}.
\end{eqnarray}
Considering the KKT condition as
\begin{eqnarray}
{\nabla _{{\mathbf{D}_k}}}{\rm{\varepsilon }} = {\mathbf{D}_k}{\mathbf{G}_k}\mathbf{W}{\mathbf{W}^*}\mathbf{G}_k^* - {\mathbf{W}^*}\mathbf{G}_k^* + \sigma _k^2{\mathbf{D}_k} = 0,
\end{eqnarray}
the optimum value for $\mathbf{D}_k$ is obtained as
\begin{eqnarray}
{\mathbf{D}_k} = {\mathbf{W}^*}\mathbf{G}_k^*{\left( {{\mathbf{G}_k}\mathbf{W}{\mathbf{W}^*}\mathbf{G}_k^* + \sigma _k^2\mathbf{I}} \right)^{ - 1}}.
\end{eqnarray}
In the second step, by substituting $\mathbf{D}_k$, computed in (43), into (41), the relay precoder matrix is calculated as follows:
\begin{eqnarray}
\mathop {\min }\limits_\mathbf{W} \tilde \varepsilon, \nonumber
\end{eqnarray}
subject to 
\begin{eqnarray}
\text{Tr}\left( {\mathbf{W}{\mathbf{W}^*}} \right) \le {P_r}.
\end{eqnarray}
The KKT conditions of this problem with respect to the relay's precoder $\mathbf{W}$ are represented as
\begin{eqnarray}
{\nabla _{{\mathbf{V}_k}}}\tilde \varepsilon  + \rho {\nabla _{{\mathbf{V}_k}}}\left\{ {\text{Tr}\left( {\mathbf{W}{\mathbf{W}^*}} \right) - {P_r}} \right\}=0
\end{eqnarray}
\begin{eqnarray}
\rho \left( {\text{Tr}\left( {\mathbf{W}{\mathbf{W}^*}} \right) - {P_r}} \right) = 0,
\end{eqnarray}
where $\rho$ denotes KKT coefficient related to the relay. From (41) and (45), we obtain
\begin{eqnarray}
\mathop \sum \limits_{k = 1}^{2K} (\mathbf{G}_k^*\mathbf{D}_k^*{\mathbf{D}_k}{\mathbf{G}_k}\mathbf{W} - \mathbf{G}_k^*\mathbf{D}_k^*) + \rho \mathbf{W}=0
\end{eqnarray}
Hence, we have
\begin{eqnarray}
\mathbf{W} = {\left(\mathop \sum \limits_{k = 1}^{2K} \mathbf{G}_k^*\mathbf{D}_k^*{\mathbf{D}_k}{\mathbf{G}_k} + \rho \mathbf{I}\right)^{ - 1}}\mathop \sum \limits_{k = 1}^{2K} \mathbf{G}_k^*\mathbf{D}_k^*,
\end{eqnarray}
where $\rho$ is determined to satisfy the second KKT condition in (46) similar to the steps taken to select ${\mu _k}$ in Subsection III.A.2. Algorithm 3 presents the broadcast phase Sum MSE precoding and filter design. 
\begin{table}[t]
  \centering
  \caption{
    Algorithm 3: Broadcast Phase Sum based MSE Precoding and Projection Filter Design}
    \vspace*{-1em}
    \begin{tabular}{l}
      \hline
    Initialize ${\mathbf{W}^{\left( 0 \right)}}$ and $\delta $
\\
Iterate
\\
\hspace*{+8pt}1.Update $\mathbf{D}_k^{\left( {j + 1} \right)}, \forall k$ for fixed ${\mathbf{W}^{\left( j \right)}}$
\\
\hspace*{+8pt}2.Update ${\mathbf{W}^{\left( {j+1} \right)}}$ with finding ${\rho ^{\left( {j + 1} \right)}}$ for fixed $\mathbf{D}_k^{\left( {j + 1} \right)}, \forall k$

\\
Until $\text{Tr}\left( {\left( {{\mathbf{W}^{\left( {j + 1} \right)}} - {\mathbf{W}^{\left( j \right)}}} \right){{\left( {{\mathbf{W}^{\left( {j + 1} \right)}} - {\mathbf{W}^{\left( j \right)}}} \right)}^*}} \right) \le \delta$ 

    \vspace*{.5em}
\\

      \hline
    \end{tabular}

\end{table}

\subsubsection{Max MSE based Precoding and Projection Filter Design}
Here, the optimization problem, which minimizes the maximum of mean squared error of each user, can be considered as
\begin{eqnarray}
\mathop {\min }\limits_{{\mathbf{D}_k}, \mathbf{W}} \mathop {\max }\limits_{k = 1, \ldots ,2K} {\tilde \varepsilon _k},\nonumber
\end{eqnarray}
subject to
\begin{eqnarray}
\text{Tr}\left( {\mathbf{W}{\mathbf{W}^*}} \right) \le {P_r},
\end{eqnarray}
where with straightforward simplifications like III.A.1, ${\tilde \varepsilon _k}$ is given by
\begin{eqnarray}
\tilde \varepsilon _k={\left| {\left| {\text{vec}\left( {{\mathbf{D}_k}{\mathbf{G}_k}\mathbf{W}} \right) - \text{vec}\left( \mathbf{I} \right)} \right|} \right|^2} + \sigma _k^2{\left| {\left| {\text{vec}\left( {{\mathbf{D}_k}} \right)} \right|} \right|^2}.
\end{eqnarray}
Similar to III.A.1, this problem can be solved by the alternative optimization method. In the first step, for $\mathbf{W}$, we consider
\begin{eqnarray}
\mathop {\min }\limits_{\{ {\mathbf{W}},x\} } x, \nonumber
\end{eqnarray}
subject to
\begin{eqnarray}
\left\{ {\begin{array}{*{20}{c}}{\left|\left|\begin{array}{*{20}{c}}{{\sigma _k}{||\text{vec}(\mathbf{D}}_k)||}\\{\left( {\mathbf{I} \otimes {\mathbf{D}_k}{\mathbf{G}_k}} \right){\rm{\text{vec}}}\left( {\rm{\mathbf{W}}} \right) - {\rm{\text{vec}}}\left( {\rm{\mathbf{I}}} \right)}\end{array}\right|\right|^2 \le x}\\{{\left| {\left| {\text{vec}\left( \mathbf{W} \right)} \right|} \right|^2} \le {P_r}, k=1,..., 2K.}\end{array}} \right. 
\end{eqnarray}
and in the second step, for $\mathbf{D}_k, \forall k$, we consider
\begin{eqnarray}
\mathop {\min }\limits_{\{ {\mathbf{D}_k};x\} } x, \nonumber
\end{eqnarray}
subject to
\begin{eqnarray}
\left\{ {\begin{array}{*{20}{c}}{\left|\left|\begin{array}{*{20}{c}}{{\sigma _k}{||\text{vec}(\mathbf{D}}_k)||}\\{\left( {{\mathbf{W}^*}{\mathbf{G}_k^*}} \otimes \mathbf{I} \right){\rm{\text{vec}}}\left( {\rm{\mathbf{D}_k}} \right) - {\rm{\text{vec}}}\left( {\rm{\mathbf{I}}} \right)}\end{array}\right|\right|^2 \le x}\end{array}} \right. .
\end{eqnarray}
The above optimization problems are SOCP. Thus, they can be solved by standard SOCP solver. However, it is clear that the answer of (52) is equal to (43). The procedure is shown in algorithm 4.
\begin{table}[t]
  \centering
  \caption{
    Algorithm 4: Broadcast Phase Max MSE based Precoding and Projection Filter Design}
    \vspace*{-1em}
    \begin{tabular}{l}
      \hline
   Initialize ${\mathbf{W}^{\left( 0 \right)}}$ and $\delta $
\\
Iterate
\\
\hspace*{+8pt}1.Update ${\mathbf{W}^{\left( {j+1} \right)}}$ by solving SOCP problem of (51)  for fixed $\mathbf{D}_k^{\left( {j} \right)}, \forall k$
\\
\hspace*{+8pt}2.Update $\mathbf{D}_k^{\left( {j + 1} \right)}, \forall k$ by solving SOCP problem of (52) for fixed ${\mathbf{W}^{\left( {j+1} \right)}}$

\\
Until $\text{Tr}\left( {\left( {{\mathbf{W}^{\left( {j + 1} \right)}} - {\mathbf{W}^{\left( j \right)}}} \right){{\left( {{\mathbf{W}^{\left( {j + 1} \right)}} - {\mathbf{W}^{\left( j \right)}}} \right)}^*}} \right) \le \delta$ 
    \vspace*{.5em}
\\

      \hline
    \end{tabular}

\end{table}

\section{Robust MSE based Precoding and Projection Filter Design for the Imperfect Channel Knowledge}
The transceiver proposed in the previous section requires perfect CSI. However, in practice, CSI is not perfect due to factors such as channel estimation error or feedback delay.  In this section, we propose a robust precoding and projection filter design for the IFF scheme with imperfect CSI. We can model the CSI error as: ${\mathbf{\hat H}_k} = {\mathbf{H}_k} + {\mathbf{e}_k}, k=1,...,2K$ and ${\mathbf{\hat G}_k} = {\mathbf{G}_k} + {\mathbf{\hat e}_k}, k=1,...,2K$, where ${\mathbf{\hat H}_k}$ and ${\mathbf{\hat G}_k}$ are estimated channel matrices from user $k$ to relay R and vice versa, respectively. In addition, ${\mathbf{e}_k}$ and ${\mathbf{\hat e}_k}$ are the estimation error matrices for the related channels. We assume the components of error matrices ${\mathbf{e}_k}$ and ${\mathbf{\hat e}_k}$ have independent Gaussian distribution with $\mathbb{E}\left\{ {{\mathbf{e}_k}\mathbf{e}_k^*} \right\} = \sigma _h^2\mathbf{I}$ and $\mathbb{E}\left\{ {{{\mathbf{\hat e}}_k}\mathbf{\hat e}_k^*} \right\} = \sigma _g^2\mathbf{I}$, respectively.
 
First, we introduce the modified IFLR. We then derive the optimum precoder and projection matrices in Subsection IV.B and Subsection IV.C.

\subsection{Modified IFLR}
After signal alignment in each pair based on estimated channels as
\begin{eqnarray}
{\mathbf{\hat H}_k}{\mathbf{V}_k} = {\mathbf{\hat H}_{\bar k}}{\mathbf{V}_{\bar k}}, k=1,...,K,
\end{eqnarray}
and therefore
\begin{eqnarray}
{\mathbf{V}_{\bar k}} = \mathbf{\hat H}_{\bar k}^\dag {\mathbf{\hat H}_k}{\mathbf{V}_k}.
\end{eqnarray}
From (2), we can write the received signal as
\begin{eqnarray}
{\mathbf{y}_r} &=& \mathop \sum \limits_{k = 1}^{2K} {\mathbf{\hat H}_k}{\mathbf{V}_k}{\mathbf{\dot s}_k} + \mathop \sum \limits_{k = 1}^K {\mathbf{e}_k}{\mathbf{V}_k}{\mathbf{s}_k} + \mathop \sum \limits_{k = 1}^K {\mathbf{e}_{\bar k}}{\mathbf{V}_{\bar k}}{\mathbf{s}_{\bar k}} + {\mathbf{z}_r}\nonumber\\&=& \mathbf{\hat H}\mathbf{\dot S} + \mathop \sum \limits_{k = 1}^K {\mathbf{e}_k}{\mathbf{V}_k}{\mathbf{s}_k} + \mathop \sum \limits_{ k = 1}^{K} {\mathbf{e}_{\bar k}}{\mathbf{V}_{\bar k}}{\mathbf{s}_{\bar k}} + {\mathbf{z}_r},
\end{eqnarray}
where
\begin{eqnarray}
\mathbf{\hat H} \buildrel \Delta \over = \left[ {{{\mathbf{\hat H}}_1}{\mathbf{V}_1}, \ldots ,{{\mathbf{\hat H}}_K}{\mathbf{V}_K}} \right].
\end{eqnarray}
Similar to the Section II, to recover an equation with ECV $\mathbf{a}_k$, $\mathbf{y}_r$ is projected onto vector $\mathbf{b}_k$, as: 
\begin{eqnarray}
\mathbf{b}_k^*{\mathbf{y}_r} = \mathbf{a}_k^*\mathbf{\dot S} + \left( {\mathbf{b}_k^*\mathbf{\hat H} - \mathbf{a}_k^*} \right)\mathbf{\dot S} + \mathbf{b}_k^*\mathop \sum \limits_{l = 1}^K {\mathbf{e}_l}{\mathbf{V}_l}{\mathbf{s}_l} + \mathbf{b}_k^*\mathop \sum \limits_{l = 1}^K {\mathbf{e}_{\bar l}}{\mathbf{V}_{\bar l}}{\mathbf{s}_{\bar l}} + \mathbf{b}_k^*{\mathbf{z}_r}.
\end{eqnarray}
Hence, the effective noise variance for this recovering is given by
\begin{eqnarray}
{\varepsilon _{e,k}} = \mathbb{E}\left\{ {{{\left| {\left| {\mathbf{b}_k^*{\mathbf{y}_r} - \mathbf{a}_k^*\mathbf{\dot S}} \right|} \right|}^2}} \right\}.
\end{eqnarray}
With some straightforward simplifications, (58) can be rewritten as
\begin{eqnarray}
{\varepsilon _{e,k}} = \mathbb{E}\left\{ {2{{\left| {\left| {{{\mathbf{\hat H}}^*}{\mathbf{b}_k} - {\mathbf{a}_k}} \right|} \right|}^2} + \sigma _h^2\mathop \sum \limits_{l = 1}^K {{\left| {\left| {\mathbf{b}_k^*\mathop \sum \limits_{l = 1}^K {\mathbf{e}_l}{\mathbf{V}_l}{\mathbf{s}_l}} \right|} \right|}^2} + \sigma _h^2\mathop \sum \limits_{l = 1}^K {{\left| {\left| {\mathbf{b}_k^*\mathop \sum \limits_{l = 1}^K {\mathbf{e}_{\bar l}}{\mathbf{V}_{\bar l}}{\mathbf{s}_{\bar l}}} \right|} \right|}^2} + \sigma _r^2{{\left| {\left| {{\mathbf{b}_k}} \right|} \right|}^2}} \right\}.
\end{eqnarray}
\begin{theorem}
By considering the error matrices $\mathbf{e}_{k}, k=1,...,K$ with $\mathbb{E}\left\{ {{\mathbf{e}_k}\mathbf{e}_k^*} \right\} = \sigma _h^2\mathbf{I}$ and $\mathbb{E}\left\{ {{\mathbf{e}_k}\mathbf{e}_{\hat k}^*} \right\} = 0, \forall k = \hat k$, messages $\mathbf{s}_l, l=1,...,K$ with $\mathbb{E}\left\{ {{\mathbf{s}_l}\mathbf{s}_l^*} \right\} = 1$ and $\mathbb{E}\left\{ {{\mathbf{s}_l}\mathbf{s}_{\hat l}^*} \right\} = 0, \forall l \ne \hat l$, matrices $\mathbf{V}_l, l=1,...,K$, and vector $\mathbf{b}_k$, we have
\begin{eqnarray}
\mathbb{E}\left\{{\left| {\left| {\mathbf{b}_k^*\mathop \sum \limits_{l = 1}^K {\mathbf{e}_l}{\mathbf{V}_l}{\mathbf{s}_l}} \right|} \right|^2}\right\} = \sigma _h^2\mathop \sum \limits_{l = 1}^K \text{Tr}\left( {\mathbf{V}_l^*{\mathbf{V}_l}} \right){\left| {\left| {{\mathbf{b}_k}} \right|} \right|^2}.
\end{eqnarray}
\end{theorem}
\begin{proof}
The proof is given in Appendix I.
\end{proof}
According to Theorem 1, the expression in (59) becomes
\begin{eqnarray}
{\varepsilon _{e,k}} = 2{\left| {\left| {{{\mathbf{\hat H}}^*}{\mathbf{b}_k} - {\mathbf{a}_k}} \right|} \right|^2} + \sigma _h^2\mathop \sum \limits_{l = 1}^K \left( {\text{Tr}\left( {\mathbf{V}_l^*{\mathbf{V}_l}} \right) + \text{Tr}\left( {\mathbf{V}_{\bar l}^*{\mathbf{V}_{\bar l}}} \right)} \right){\left| {\left| {{\mathbf{b}_k}} \right|} \right|^2} + \sigma _r^2{\left| {\left| {{\mathbf{b}_k}} \right|} \right|^2}.
\end{eqnarray}
Accordingly, the computation rate for the equation with ECV $\mathbf{a}_k$ is given by
\begin{eqnarray}
{R_k} = {\rm{lo}}{{\rm{g}}^ + }\left( {\frac{1}{{2{{\left| {\left| {{{\mathbf{\hat H}}^*}{\mathbf{b}_k} - {\mathbf{a}_k}} \right|} \right|}^2} + \sigma _h^2\mathop \sum \nolimits_{l = 1}^K \left( {\text{Tr}\left( {\mathbf{V}_l^*{\mathbf{V}_l}} \right) + \text{Tr}\left( {\mathbf{V}_{\bar l}^*{\mathbf{V}_{\bar l}}} \right)} \right){{\left| {\left| {{\mathbf{b}_k}} \right|} \right|}^2} + \sigma _r^2{{\left| {\left| {{\mathbf{b}_k}} \right|} \right|}^2}}}} \right).
\end{eqnarray}
Note that an equation with message transmission power $P$ and effective recovery noise variance $N$ has computation rate ${\rm{lo}}{{\rm{g}}^ + }\left( {\frac{P}{{{N}}}} \right)$ [6].
\begin{theorem}
The optimum projection vector $\mathbf{b}_k$ for recovering the equation with ECV $\mathbf{a}_k$ is     
\begin{eqnarray}
\mathbf{b}_k^* = \mathbf{a}_k^*{\mathbf{\hat H}^*}{\left( {\frac{{\sigma _r^2}}{2}\mathbf{I} + \frac{{\sigma _h^2}}{2}\left( {\mathop \sum \limits_{l = 1}^K \text{Tr}\left( {\mathbf{V}_l^*{\mathbf{V}_l}} \right) + \text{Tr}\left( {\mathbf{V}_{\bar l}^*{\mathbf{V}_{\bar l}}} \right)} \right)\mathbf{I} + \mathbf{\hat H}{{\mathbf{\hat H}}^*}} \right)^{ - 1}}\mathbf{\hat H},
\end{eqnarray}
and hence, the projection matrix $\mathbf{B}$ becomes
\begin{eqnarray}
\mathbf{B} = \mathbf{A}{\mathbf{\hat H}^*}{\left( {\frac{{\sigma _r^2}}{2}\mathbf{I} + \frac{{\sigma _h^2}}{2}\left( {\mathop \sum \limits_{l = 1}^K \text{Tr}\left( {\mathbf{V}_l^*{\mathbf{V}_l}} \right) + \text{Tr}\left( {\mathbf{V}_{\bar l}^*{\mathbf{V}_{\bar l}}} \right)} \right)\mathbf{I} + \mathbf{\hat H}{{\mathbf{\hat H}}^*}} \right)^{ - 1}}\mathbf{\hat H}.
\end{eqnarray}
\end{theorem}
\begin{proof}
The proof is given in Appendix II.
\end{proof}
By substituting (63) into (61) and some straightforward simplifications, the effective noise variance $\varepsilon _{e,k}$ is obtained as  
\begin{eqnarray}
\varepsilon _{e,k} = \mathbf{a}_k^*\mathbf{U}\mathbf{a}_k,
\end{eqnarray}
where
\begin{eqnarray}
\mathbf{U} = \mathbf{I} - {\mathbf{\hat H}^*}{\left( {\frac{{\sigma _r^2}}{2}\mathbf{I} + \frac{{\sigma _h^2}}{2}\left( {\mathop \sum \limits_{l = 1}^K \text{Tr}\left( {\mathbf{V}_l^*{\mathbf{V}_l}} \right) + \text{Tr}\left( {\mathbf{V}_{\bar l}^*{\mathbf{V}_{\bar l}}} \right)} \right)\mathbf{I} + \mathbf{\hat H}{{\mathbf{\hat H}}^*}} \right)^{ - 1}}\mathbf{\hat H}.
\end{eqnarray}       
The other concepts by replacing (10) with (63) are similar to Section II.

\subsection{Robust Multiple-Access Phase based MSE Precoding and Projection Filter Design}
Here, we consider Sum-Equation MSE and Max-Equation MSE critera for transceiver design with imperfect CSI. The problems (22) and (32) get solved by considering the new $\mathbf{U}$ in (66). 
\subsubsection{Robust Sum-Equation MSE based Precoding and Projection Filter Design}
From (57) and (58), the Sum-Equation MSE minimization problem considering the estimated channel matrix ${\mathbf{\hat H}_k}$ can be written as
\begin{eqnarray}
\mathop {\min }\limits_{{\mathbf{V}_k},\mathbf{A},\mathbf{B}} \varepsilon_e  =\mathop \sum \limits_{i = 1}^L\varepsilon_{e,k} = \mathop \sum \limits_{i = 1}^L \left. \mathbb{E}\left\{ {{\left| {\left| {\mathbf{b}_i^*{\mathbf{y}_r} - \mathbf{a}_i^*\mathbf{\dot S}} \right|} \right|}^2} \right|{\mathbf{\hat H}_i}\right\},\nonumber
\end{eqnarray}
subject to  
\begin{eqnarray}
\text{Tr}\left( {{\mathbf{V}_k}\mathbf{V}_k^*} \right) + \text{Tr}\left( {\mathbf{\hat H}_{\bar k}^\dag {{\mathbf{\hat H}}_k}{\mathbf{V}_k}\mathbf{V}_k^*\mathbf{\hat H}_k^*\mathbf{\hat H}{{_{\bar k}^\dag }^*}} \right) \le {P_k}.
\end{eqnarray}
The objective function in (67) can be simplified to
\begin{eqnarray}
\varepsilon_e &=& \mathop \sum \limits_{k = 1}^L \mathop \sum \limits_{i=1}^K 2(\mathbf{b}_k^*{\mathbf{\hat H}_i}{\mathbf{V}_i}\mathbf{V}_i^*\mathbf{\hat H}_i^*{\mathbf{b}_k} - 2\mathbf{a}_{k,i}^*\mathbf{V}_i^*\mathbf{\hat H}_i^*{\mathbf{b}_k} + \mathbf{a}_{k,i}^*{\mathbf{a}_{k,i}}) \nonumber\\ &+& \sigma _h^2\mathop \sum \limits_{l = 1}^K \left( {\text{Tr}\left( {\mathbf{V}_l^*{\mathbf{V}_l}} \right) + \text{Tr}\left( {\mathbf{V}_{\bar l}^*{\mathbf{V}_{\bar l}}} \right)} \right)\mathbf{b}_k^*{\mathbf{b}_k} + \sigma _r^2\mathbf{b}_k^*{\mathbf{b}_k}.
\end{eqnarray}
Similar to the procedure of Subsection III.A.2, using KKT conditions, we have
\begin{eqnarray}
2\mathop \sum \limits_{i = 1}^L (\mathbf{\hat H}_k^*{\mathbf{b}_i}\mathbf{b}_i^*{\mathbf{\hat H}_k}{\mathbf{V}_k} - \mathbf{\hat H}_k^*{\mathbf{b}_i}\mathbf{a}_{k,i}^* &+& \sigma _h^2\left( {{\mathbf{V}_k} + \mathbf{\hat H}_k^*\mathbf{\hat H}{{_{\bar k}^\dag }^*}\mathbf{\hat H}_{\bar k}^\dag {\mathbf{\hat H}_k}{\mathbf{V}_k}} \right)\mathbf{b}_i^*{\mathbf{b}_i}) \nonumber\\ &+& {\mu _k}\left( {{\mathbf{V}_k} + \mathbf{\hat H}_k^*\mathbf{\hat H}{{_{\bar k}^\dag }^*}\mathbf{\hat H}_{\bar k}^\dag {\mathbf{\hat H}_k}{\mathbf{V}_k}} \right)=0.
\end{eqnarray}
Thus, we have
\begin{eqnarray}
{\mathbf{V}_k} = {(\mathbf{\hat H}_k^*\mathop \sum \limits_{i = 1}^L {\mathbf{b}_i}\mathbf{b}_i^*{\mathbf{\hat H}_k} + \left( {\frac{1}{2}{\mu _k} + \frac{{\sigma _h^2}}{2}\mathop \sum \limits_{i = 1}^L \mathbf{b}_i^*{\mathbf{b}_i}} \right)\left( {\mathbf{I} + \mathbf{\hat H}_k^*\mathbf{\hat H}{{_{\bar k}^\dag }^*}\mathbf{\hat H}_{\bar k}^\dag {{\mathbf{\hat H}}_k}} \right))^{ - 1}}\mathop \sum \limits_{i = 1}^L \mathbf{\hat H}_k^*{\mathbf{b}_i}\mathbf{a}_{k,i}^*,
\end{eqnarray}
\begin{eqnarray}
{\mu _k}\left( {\text{Tr}\left( {{\mathbf{V}_k}\mathbf{V}_k^*} \right) + \text{Tr}\left( {\mathbf{\hat H}_{\bar k}^\dag {{\mathbf{\hat H}}_k}{\mathbf{V}_k}\mathbf{V}_k^*\mathbf{\hat H}_k^*\mathbf{\hat H}{{_{\bar k}^\dag }^*}} \right) - {P_k}} \right) = 0, \forall k=1,...,K.
\end{eqnarray}  
The parameter $\mu _k$ can be obtained as proposed in Subsection III.A.2. Algorithm 2 can be used by replacing (36) with (70).
 
\subsubsection{Robust Max-Equation MSE based Precoding and Projection Filter Design}
The ${\varepsilon _{e,k}}$ can be written as
\begin{eqnarray}
{\varepsilon _{e,k}}&=&\mathop \sum \limits_{i = 1}^K 2{ || {\mathbf{V}_i^*\mathbf{\hat H}_i^*{\mathbf{b}_k} - {\mathbf{a}_{k,i}}} ||^2} +\sigma _h^2\left(\text{Tr}\left( {\mathop \sum \limits_{l = 1}^K \mathbf{V}_l^*{\mathbf{V}_l}} \right) + \text{Tr}\left( {\mathop \sum \limits_{l = 1}^K \mathbf{V}_l^*\mathbf{\hat H}_l^*\mathbf{\hat H}{{_{\bar l}^\dag }^*}\mathbf{\hat H}_{\bar l}^\dag {{\mathbf{\hat H}}_l}{\mathbf{V}_l}} \right)\right){\left| {\left| {{\mathbf{b}_k}} \right|} \right|^2}\nonumber\\&+& \sigma _r^2{\left| {\left| {{\mathbf{b}_k}} \right|} \right|^2}=\left| {\left| {{\mathbb{V}^*}{\mathbb{\hat H}^*}{\mathbf{b}_k} - {\mathbf{a}_k}} \right|} \right|^2 +\sigma _h^2{\left| {\left| {{\mathbf{b}_k}} \right|} \right|^2}\left({\rm{\text{Tr}}}\left( {{{\rm{\mathbb{V}}}^{\rm{*}}}{\rm{\mathbb{V}}}} \right) + {\rm{\text{Tr}}}\left( {{{\rm{\mathbb{V}}}^{\rm{*}}}{\rm{\mathbf{\Phi}^* \mathbf{\Phi}\mathbb{V}}}} \right)\right)+ \sigma _r^2{\left| {\left| {{\mathbf{b}_k}} \right|} \right|^2}\nonumber\\&=&\left|\left|\begin{array}{*{20}{c}}{{\sigma _r}{||\mathbf{b}_k||}}\\{\left( {\mathbf{b}_k^*\mathbb{\hat H} \otimes \mathbf{I}} \right)\text{vec}\left( {{\mathbb{V}^*}} \right) - {\mathbf{a}_k}}\\{{\sigma _h}{||\mathbf{b}_k||}\left( {\text{vec}\left( {{\mathbb{V}^*}} \right) + \text{vec}\left( {{\mathbb{V}^*}{{\mathbf{\Phi}}^*} } \right)} \right)}\end{array}\right|\right|^2.
\end{eqnarray}       
where
\begin{eqnarray}
\mathbb{\hat H} \buildrel \Delta \over = \left[ {\begin{array}{*{20}{c}}{{\mathbf{\hat H}_1}}&{\begin{array}{*{20}{c}} \ldots &{{\mathbf{\hat H}_K}}\end{array}}\end{array}} \right],
\end{eqnarray}
\begin{eqnarray}
\mathbf{\Phi}\buildrel \Delta \over =\left( {\begin{array}{*{20}{c}}{\mathbf{\hat H}_{\bar 1}^\dag {{\mathbf{\hat H}}_1}}&{\begin{array}{*{20}{c}}0& \cdots \end{array}}&0\\{\begin{array}{*{20}{c}}0\\ \vdots \end{array}}& \ddots &{\begin{array}{*{20}{c}} \vdots \\0\end{array}}\\0&{\begin{array}{*{20}{c}} \cdots &0\end{array}}&{\mathbf{\hat H}_{\bar K}^\dag {{\mathbf{\hat H}}_K}}\end{array}} \right).
\end{eqnarray}
Similar to Subsection III.A.1, the optimization problem of transmit precoding matrices $\mathbf{V}_k, \forall k$ can be written as
\begin{eqnarray}
\mathop {\min }\limits_{\{ {\mathbf{V}_k};x\} } x, \nonumber
\end{eqnarray}
subject to
\begin{eqnarray}
\left\{ {\begin{array}{*{20}{c}}{\left|\left|\begin{array}{*{20}{c}}{{\sigma _r}{||\mathbf{b}_k||}}\\{\left( {\mathbf{b}_k^*\mathbb{\hat H} \otimes \mathbf{I}} \right)\text{vec}\left( {{\mathbb{V}^*}} \right) - {\mathbf{a}_k}}\\{{\sigma _h}{||\mathbf{b}_k||}\left( {\text{vec}\left( {{\mathbb{V}^*}} \right) + \text{vec}\left( {{\mathbb{V}^*}{{\mathbf{\Phi}}^*} } \right)} \right)}\end{array}\right|\right|^2 \le x}\\{{\left| {\left| {{\rm{\text{vec}}}\left( {{\mathbf{V}_k}} \right)} \right|} \right|^2} + {\left| {\left| {\text{vec}\left( {{\rm{ \mathbf{H}_{\bar k}^\dag {\mathbf{H}_k}{\mathbf{V}_k}}}} \right)} \right|} \right|^2} \le {P_k}, k=1,..., K.}\end{array}} \right.
\end{eqnarray}
Similarly, the above optimization problem is a SOCP problem, and algorithm 1 can be used by replacing (28) with (75).

\subsection{Robust Broadcast Phase based MSE Precoding and Projection Filter Design}
Here, in the second phase, we consider Sum MSE and Max MSE with imperfect CSI.
\subsubsection{Robust Sum MSE based Precoding and Projection Filter Design}
The minimization problem defined in (38) and (39), considering the estimated channel matrix ${\mathbf{\hat G}_k}$ can be modified to
\begin{eqnarray}
\mathop {\min }\limits_{\mathbf{W},{\mathbf{D}_k}} {\tilde \varepsilon_e}  = \mathop \sum \limits_{k = 1}^{2K} {\tilde \varepsilon _{e,k}}= \mathop \sum \limits_{k = 1}^{2K} \left. \mathbb{E}\left\{ {{\left| {\left| {{\mathbf{D}_k}{\mathbf{y}_k} - \mathbf{t}} \right|} \right|^2}}\right|{\mathbf{\hat G}_k}\right\},\nonumber
\end{eqnarray}
subject to
\begin{eqnarray}
\text{Tr}\left( {\mathbf{W}{\mathbf{W}^*}} \right) \le {P_r},
\end{eqnarray}
where
\begin{eqnarray}
\tilde \varepsilon_e  = \mathop \sum \limits_{k = 1}^{2K} \text{Tr}\left\{ {{\mathbf{D}_k}{{\mathbf{\hat G}}_k}\mathbf{W}{\mathbf{W}^*}\mathbf{\hat G}_k^*\mathbf{D}_k^* - 2{\mathbf{W}^*}\mathbf{\hat G}_k^*\mathbf{D}_k^* + \sigma _g^2\text{Tr}\left( {{\mathbf{W}^*}\mathbf{W}} \right){\mathbf{D}_k}\mathbf{D}_k^* + \mathbf{I} + \sigma _u^2{\mathbf{D}_k}\mathbf{D}_k^*} \right\}.
\end{eqnarray}
To solve the problem, with KKT conditions similar to the solution of the problem presented in Subsection III.B.1, we have
\begin{eqnarray}
{\mathbf{D}_k} = {\mathbf{W}^*}\mathbf{\hat G}_k^*{\left( {{{\mathbf{\hat G}}_k}\mathbf{W}{\mathbf{W}^*}\mathbf{\hat G}_k^* + \sigma _g^2\text{Tr}\left( {{\mathbf{W}^*}\mathbf{W}} \right)\mathbf{I} + \sigma _u^2\mathbf{I}} \right)^{ - 1}}.
\end{eqnarray}
Moreover, to find $\mathbf{W}$, according to the KKT condition in (45), we can write
\begin{eqnarray}
\mathop \sum \limits_{k = 1}^{2K} (\mathbf{\hat G}_k^*\mathbf{D}_k^*{\mathbf{D}_k}{\mathbf{\hat G}_k}\mathbf{W} + \sigma _g^2\mathbf{W}\text{Tr}({\mathbf{D}_k{\mathbf{D}_k^*}}) - \mathbf{\hat G}_k^*\mathbf{D}_k^*) + \rho \mathbf{W}=0.
\end{eqnarray}
Hence, we have 
\begin{eqnarray}
\mathbf{W} = {\left(\mathop \sum \limits_{k = 1}^{2K} \mathbf{\hat G}_k^*\mathbf{D}_k^*{\mathbf{D}_k}{\mathbf{\hat G}_k} + \sigma _g^2\mathop \sum \limits_{k = 1}^{2K}\text{Tr}({\mathbf{D}_k{\mathbf{D}_k^*}})+ \rho \mathbf{I}\right)^{ - 1}}\mathop \sum \limits_{k = 1}^{2K} \mathbf{\hat G}_k^*\mathbf{D}_k^*,
\end{eqnarray}
\begin{eqnarray}
\rho \left( {\text{Tr}\left( {\mathbf{W}{\mathbf{W}^*}} \right) - {P_r}} \right) = 0.
\end{eqnarray}
The parameter $\rho$ can be obtained similar to what explained in Subsection III.B.1. Algorithm 3 can be used by replacing (43) and (48) with (78) and (80), respectively.
\subsubsection{Robust Max MSE based Precoding and Projection Filter Design}
We consider the following optimization problem:
\begin{eqnarray}
\mathop {\min }\limits_{{\mathbf{D}_k}, \mathbf{W}} \mathop {\max }\limits_{k = 1, \ldots ,2K} {\tilde \varepsilon _{e,k}},\nonumber
\end{eqnarray}
subject to
\begin{eqnarray}
\text{Tr}\left( {\mathbf{W}{\mathbf{W}^*}} \right) \le {P_r},
\end{eqnarray}
where from (77), the ${\tilde \varepsilon _{e,k}}$ is given by
\begin{eqnarray}
\tilde \varepsilon _{e,k}={\left| {\left| {\text{vec}\left( {{\mathbf{D}_k}{\mathbf{\hat G}_k}\mathbf{W}} \right) - \text{vec}\left( \mathbf{I} \right)} \right|} \right|^2} + \sigma _g^2{\left| {\left| {\text{vec}\left( \mathbf{W} \right)} \right|} \right|^2}{\left| {\left| {\text{vec}\left( {{\mathbf{D}_k}} \right)} \right|} \right|^2}+\sigma _k^2{\left| {\left| {\text{vec}\left( {{\mathbf{D}_k}} \right)} \right|} \right|^2}.
\end{eqnarray}
This problem can be solved by the alternative optimization method. In the first step, for $\mathbf{W}$, we consider
\begin{eqnarray}
\mathop {\min }\limits_{\{ {\mathbf{W}};x\} } x, \nonumber
\end{eqnarray}
subject to
\begin{eqnarray}
\left\{ {\begin{array}{*{20}{c}}{\left|\left|\begin{array}{*{20}{c}}{{\sigma _k}{||\text{vec}(\mathbf{D}}_k)||}\\{\left( {\mathbf{I} \otimes {\mathbf{D}_k}{\mathbf{\hat G}_k}} \right){\rm{\text{vec}}}\left( {\rm{\mathbf{W}}} \right) - {\rm{\text{vec}}}\left( {\rm{\mathbf{I}}} \right)}\\{\sigma_g ||\text{vec}(\mathbf{D}_k)||\text{vec}(\mathbf{W})}\end{array}\right|\right|^2 \le x}\\{{\left| {\left| {\text{vec}\left( \mathbf{W} \right)} \right|} \right|^2} \le {P_r}, k=1,..., 2K.}\end{array}} \right. 
\end{eqnarray}
In the second step, for $\mathbf{D}_k, \forall k$, we consider
\begin{eqnarray}
\mathop {\min }\limits_{\{ {\mathbf{D}_k};x\} } x, \nonumber
\end{eqnarray}
subject to
\begin{eqnarray}
\left\{ {\begin{array}{*{20}{c}}{\left|\left|\begin{array}{*{20}{c}}{{\sigma _k}||\text{vec}(\mathbf{D}_k)||}\\{\left( {{\mathbf{W}^*}{\mathbf{\hat G}_k^*}} \otimes \mathbf{I} \right){\rm{\text{vec}}}\left( {\rm{\mathbf{D}_k}} \right) - {\rm{\text{vec}}}\left( {\rm{\mathbf{I}}} \right)}\\{{\sigma_g ||\text{vec}(\mathbf{W})||\text{vec}(\mathbf{D}_k)}}\end{array}\right|\right|^2 \le x}\end{array}} \right.
\end{eqnarray}
Similarily, the above optimization problems are SOCP. Algorithm 4 can be used by replacing (51) and (52) with (84) and (85), respectively.

\section{Simulation Results}
In this section, we evaluate the performance of our proposed schemes and compare the results with the existing work in the literature. For simulation evaluation, we consider a two-pair two-way system, i.e., $K=2$. The Rayleigh channel parameters are equal to $\sigma _k^2 = 1, k = 1, \ldots ,4$. The channel noises are assumed to have a unit variance, i.e. $\sigma _r^2 = \sigma _u^2 = 1$. The parameter $\delta$ in the algorithms is set to ${10^{ - 3}}$, and the target rate $R_t = 1 $bit/channel use is considered. 

Fig. 2 shows the MSE distribution among equations and the total MSE for the proposed Sum-Equation MSE Minimization scheme and Max-Equation MSE scheme, for the case that each node has two antennas, i.e. $N_r=N_k=2, k=1,...,4$, considering perfect CSI. In this Fig., for simplicity, we suppose that each user sends only one message. Hence, the relay has to recover two independent equations according to the proposed algorithms. We can see that the proposed Sum-Equation MSE minimization scheme achieves the minimum total MSE, i.e. the sum of the MSE of the equations, while the proposed Max-Equation MSE scheme has less MSE for the worst equation, which has lower rate. Fig. 3 shows the average number of cases that each user utilizes only one of the two transmitted equations of the relay. As observed, this average is decreased by the increase of the SNR, which indicates that at high SNR using all of the transmitted equations can be more beneficial to each user. Hence, since the users recover their messages by using all of the transmitted equations with a probability higher than 0.6, we expect the Max-Equation MSE, which guarantees the MSE of the worst equation among all of the equations, to have a better performance than the Sum-Equation MSE.

Fig. 4 compares the outage probability of our proposed scheme in the case of perfect CSI with the ones introduced in [10] that uses AF relaying and in [12] that uses DF relaying, i.e. Denoise-and-Forward, for $N_r=N_k=2, k=1,...,4$. As it is observed, the proposed scheme has better performance in all SNRs, and provides at least 1 dB SNR improvement in comparison with the best conventional relaying scheme. In addition, the Max based MSE precoding and filter design, using Max-Equation MSE and Max MSE, performs better compared to the Sum based MSE precoding and filter design, using Sum-Equation MSE and Sum MSE. This result justifies what we expected form Fig. 3. Note, as has been discussed before, the Max based MSE has more complexity than the Sum based MSE due to the ECV search problem. 

In Fig. 5, the average sum rate of the proposed scheme is compared with the conventional precoding and filter designs considering the availability of perfect CSI for $N_r=N_k=2, k=1,...,4$. It can be observed that our proposed scheme performs significantly better than the conventional strategies in all SNRs. For example, in sum rate of 7 bit/channel use, the proposed scheme has 1.5 dB improvement in comparison with the best conventional relaying scheme. Moreover, the Max based MSE design outperforms the Sum based MSE transceiver. The results of Fig. 4 and 5 demonstrate that the use of the interference in terms of equations has significant superiority than when the interference is considered as an additional noise, like in the conventional AF and DF schemes.

In Fig 6, the effect of the number of antennas $N$, i.e. $N_r=N_k=N, k=1,...,4$, on the performance of the system is assessed. As can be observed and expected, the sum rate of the proposed scheme increases by higher $N$. For example, in sum rate of 5 bit/channel use, the system with $N=2$ performs 4.5 dB better than the one with $N=1$. 

In Fig. 7, we investigate the effect of channel estimation errors on the performance of the system with $N_r=N_k=2, k=1,...,4$, where the error power is $\sigma _h^2=\sigma _g^2$. The plots are provided for two precoder and filter designs, the non-robust design neglecting the presence of CSI error, and the robust design. As expected, the robust design has a better performance than the non-robust design, and the improvement becomes more by increasing the error power. For instance, when error power is 0.1, the robust design performs 2 dB better in sum rate of 5 bit/channel use, and at error power 0.4, about 4 dB better in sum rate of 4 bit/channel use. Also, as can be observed, as the error power goes up, the performance is degraded even in the robust design case. For example in sum rate of 6 bit/channel use, the design with perfect CSI has 2.5 dB better performance in comparison with the robust design when there is an imperfect CSI with error power 0.1, and the robust design with error power 0.1 performs significantly better than the one with error power 0.4. In addition, the Max based MSE design performs better in different error powers.     

\section{Conclusion}
In this paper, we have proposed Integer Forcing-and-Forward scheme for the MIMO multi-pair two-way relaying system based on the integer forcing linear receiver structure. We designed the precoder and projection matrices using the proposed Equation based MSE critera, i.e. Sum-Equation MSE and Max-Equation MSE in the multiple-access phase, and conventional user based MSE critera, i.e. Sum MSE and Max MSE in the broadcast phase. We also derived the precoder and filters design at the presence of CSI error. We have introduced modified integer forcing linear receiver to overcome the channel estimation error efficiently. For the schemes, we have proposed algorithms in which the alternative method is applied, and thus, the optimum solution can be achieved. The proposed scheme shows a significantly better performance, in terms of the sum rate and the outage probability, in comparison with conventional designs. Moreover, in the case of imperfect CSI, the proposed robust transceiver design improves the system performance compared with the non-robust design, in which the effect of channel estimation error is neglected.
\begin{figure}[p]
\centering
\center
\includegraphics[width =6in]{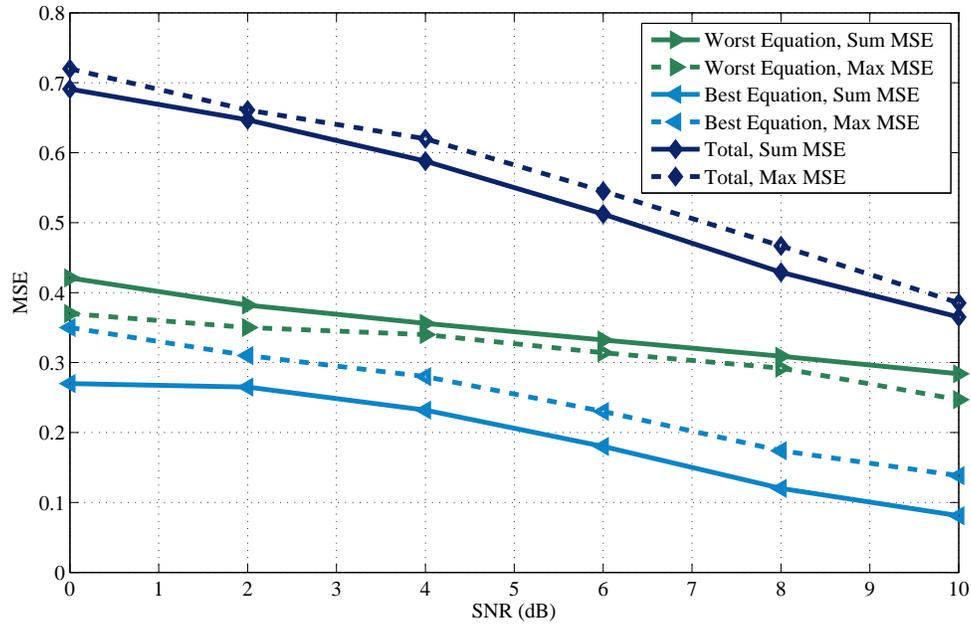}
\caption{The MSE of the proposed Max-Equation MSE and Sum-Equation MSE in a network with $K=2$ and $N_r=N_k=2, \forall k$.}
\end{figure}
\begin{figure}[p]
\centering
\center
\includegraphics[width =6in]{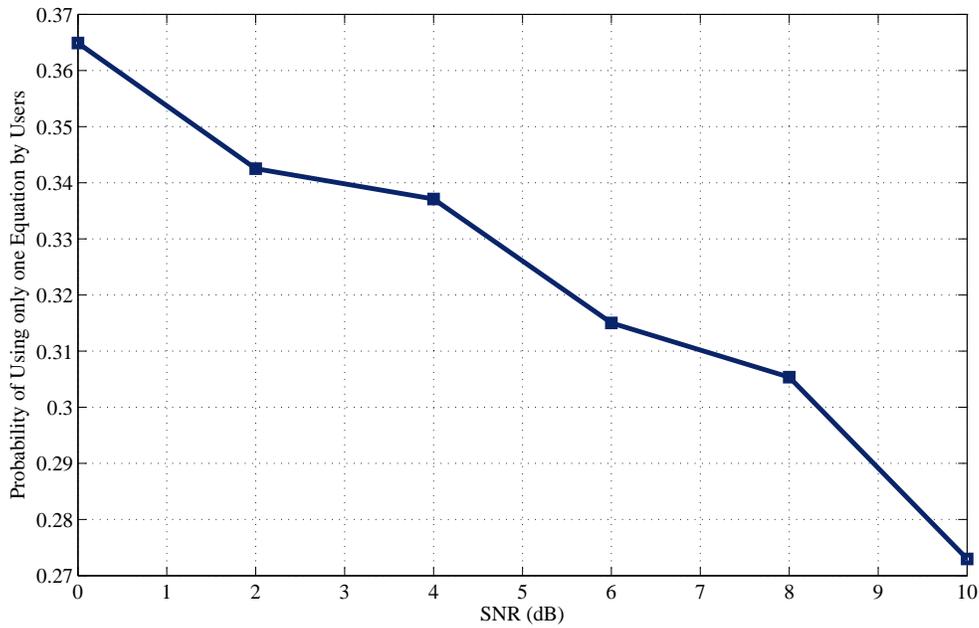}
\caption{The probability of using only one equstion by the users in a network with $K=2$ and $N_r=N_k=2, \forall k$.}
\end{figure}
\begin{figure}[p]
\centering
\center
\includegraphics[width =6in]{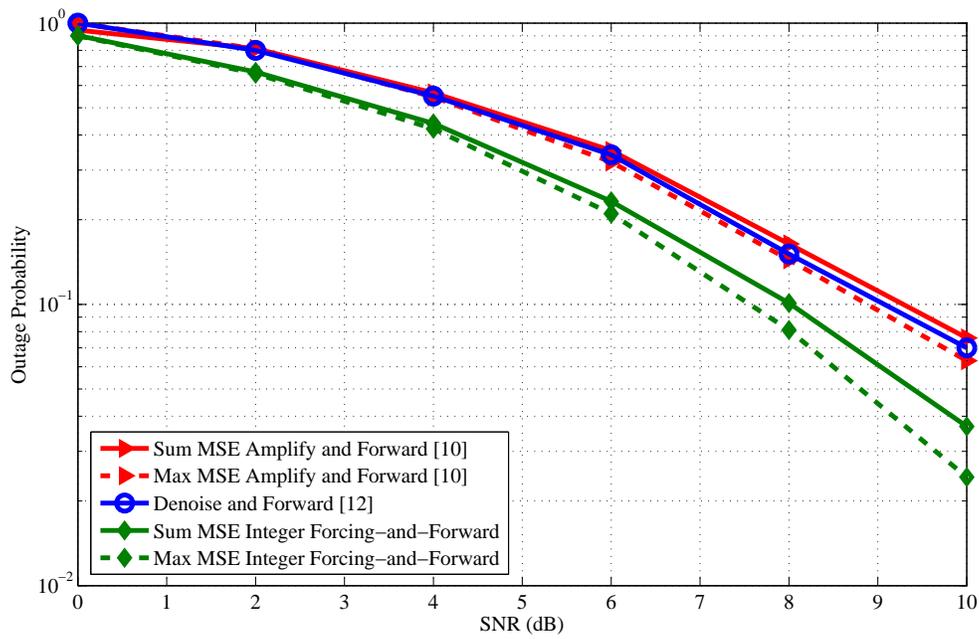}
\caption{The outage probability of the proposed scheme in comparison with conventional schemes in a network with $K=2$ and $N_r=N_k=2, \forall k$ and $R_t = 1 $bit/channel use.}
\end{figure}


\begin{figure}[p]
\centering
\center
\includegraphics[width =6in]{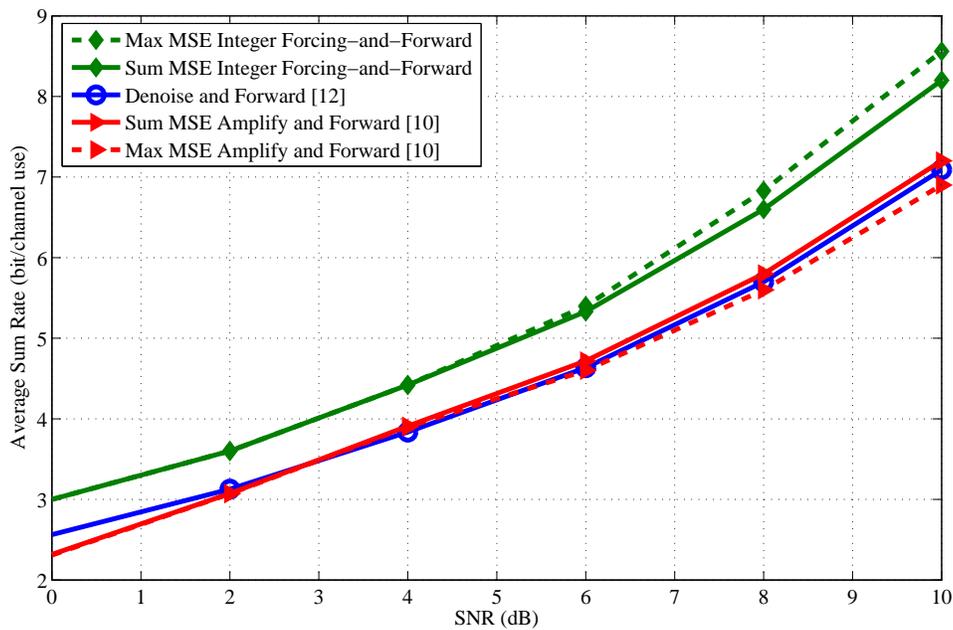}
\caption{The average sum rate of the proposed scheme in comparison with conventional schemes in a network with $K=2$ and $N_r=N_k=2, \forall k$.}

\end{figure}


\begin{figure}[p]
\centering
\center
\includegraphics[width =6in]{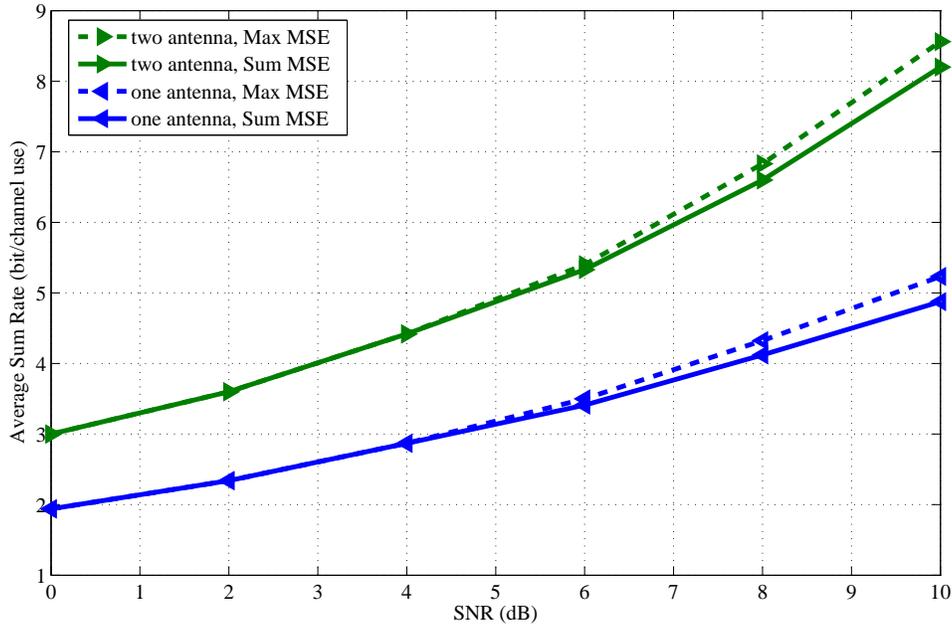}
\caption{The average sum rate of the proposed scheme with N antennas on each node, i.e. $N_r=N_k=N, \forall k$, in a network with $K=2$.}

\end{figure}


\begin{figure}[p]
\centering
\center
\includegraphics[width =6in]{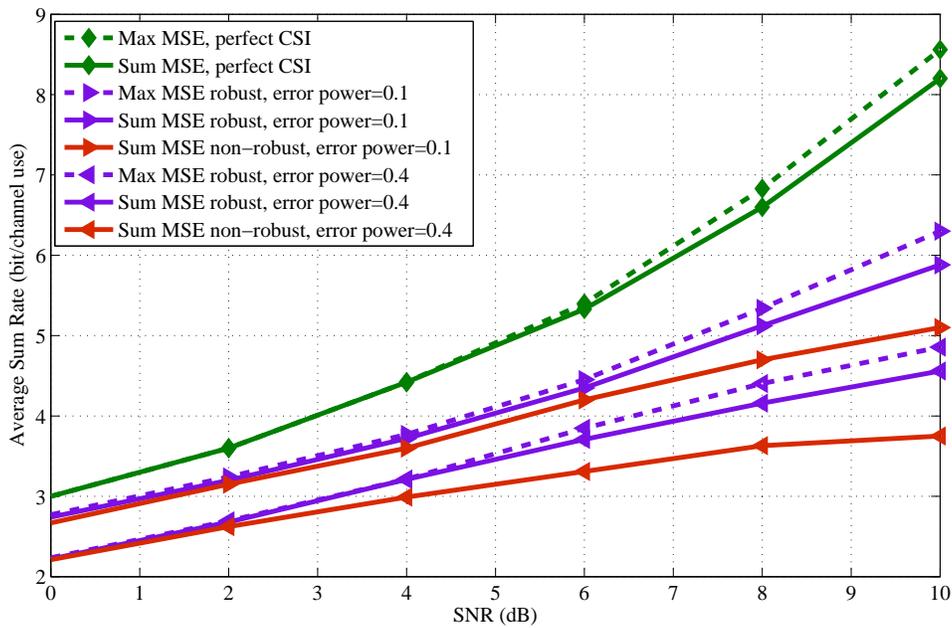}
\caption{The average sum rate of the proposed scheme for the robust and non-robust design in a network with $K=2$ and $N_r=N_k=2, \forall k$.}

\end{figure}
\newpage
\appendices

\section{Proof of Theorem 3}
With expanding, we have
\begin{eqnarray}
\mathbb{E}\left\{{{{\left| {\left| {\mathbf{b}_k^*\mathop \sum \limits_{k = 1}^K {\mathbf{e}_k}{\mathbf{V}_k}{\mathbf{s}_k}} \right|} \right|}^2}}\right\}&=& \mathbb{E}\left\{ {\left( {\mathbf{b}_k^*\mathop \sum \limits_{l = 1}^K {\mathbf{e}_l}{\mathbf{V}_l}{\mathbf{s}_l}} \right){{\left( {\mathbf{b}_k^*\mathop \sum \limits_{l = 1}^K {\mathbf{e}_l}{\mathbf{V}_l}{\mathbf{s}_l}} \right)}^*}} \right\} \nonumber\\ &=& \mathbb{E}\left\{ \mathbf{b}_k^*\mathop \sum \limits_{l = 1}^K {\mathbf{e}_l}{\mathbf{V}_l}{\mathbf{s}_l}\mathop \sum \limits_{l = 1}^K \mathbf{s}_l^*\mathbf{V}_l^*\mathbf{e}_l^*{\mathbf{b}_k}\right\}.
\end{eqnarray}
Since $\mathbb{E}\left\{ {{\mathbf{s}_l}\mathbf{s}_l^*} \right\} = 1$ and $\mathbb{E}\left\{ {{\mathbf{s}_l}\mathbf{s}_{\hat l}^*} \right\} = 0, \forall l \ne \hat l$, this leads to
\begin{eqnarray}
\mathbb{E}\left\{{{{\left| {\left| {\mathbf{b}_k^*\mathop \sum \limits_{k = 1}^K {\mathbf{e}_k}{\mathbf{V}_k}{\mathbf{s}_k}} \right|} \right|}^2}}\right\} = \mathbb{E}\left\{ {\mathbf{b}_k^*\mathop \sum \limits_{l = 1}^K {\mathbf{e}_l}{\mathbf{V}_l}\mathbf{V}_l^*\mathbf{e}_l^*{\mathbf{b}_k}} \right\}.
\end{eqnarray}
On the other hand, for any random vector $\mathbf{x}$ with mean $\mathbf{m}$ and covariance $\mathbf{M}$, and an matrix $\mathbf{A}$, we have [17]:
\begin{eqnarray}
\mathbb{E}\left\{ {{\mathbf{x}^*}{\mathbf{A}^*}\mathbf{A}\mathbf{x}} \right\} = \text{Tr}\left( {\mathbf{A}\mathbf{M}{\mathbf{A}^*}} \right) + {\mathbf{m}^*}{\mathbf{A}^*}\mathbf{A}\mathbf{m}.
\end{eqnarray}
From (88) and the fact that $\mathbb{E}\left\{ {{\mathbf{e}_k}\mathbf{e}_k^*} \right\} = \sigma _h^2\mathbf{I}$ and $\mathbb{E}\left\{ {{\mathbf{e}_k}\mathbf{e}_{\hat k}^*} \right\} = 0, \forall k = \hat k$, we have
\begin{eqnarray}
\mathbb{E}\left\{ {{\mathbf{e}_l}{\mathbf{V}_l}\mathbf{V}_l^*\mathbf{e}_l^*} \right\} = \sigma _h^2\text{Tr}\left( {\mathbf{V}_l^*{\mathbf{V}_l}} \right)\mathbf{I}.
\end{eqnarray}
And therefore, this leads to
\begin{eqnarray}
\mathbb{E}\left\{ {\mathbf{b}_k^*\mathop \sum \limits_{l = 1}^K {\mathbf{e}_l}{\mathbf{V}_l}\mathbf{V}_l^*\mathbf{e}_l^*{\mathbf{b}_k}} \right\} = \mathbf{b}_k^*\left( {\sigma _h^2\mathop \sum \limits_{l = 1}^K \text{Tr}\left( {\mathbf{V}_l^*{\mathbf{V}_l}} \right)\mathbf{I}} \right){\mathbf{b}_k} = \sigma _h^2\mathop \sum \limits_{l = 1}^K \text{Tr}\left( {\mathbf{V}_l^*{\mathbf{V}_l}} \right){\left| {\left| {{\mathbf{b}_k}} \right|} \right|^2}.
\end{eqnarray}
So, the theorem is proved.

\section{Proof of Theorem 4}

From (62), the optimum value of $\mathbf{b}_k$ is obtained by minimizing the following function:
\begin{eqnarray}
f\left( {{\mathbf{b}_k},{\mathbf{a}_k}} \right) &=& {\left| {\left| {{{\mathbf{\hat H}}^*}{\mathbf{b}_k} - {\mathbf{a}_k}} \right|} \right|^2} + \frac{{\sigma _h^2}}{2}\mathop \sum \limits_{l = 1}^K \left( {\text{Tr}\left( {\mathbf{V}_l^*{\mathbf{V}_l}} \right) + \text{Tr}\left( {\mathbf{V}_{\bar l}^*{\mathbf{V}_{\bar l}}} \right)} \right){\left| {\left| {{\mathbf{b}_k}} \right|} \right|^2} + \frac{{\sigma _r^2}}{2}{\left| {\left| {{\mathbf{b}_k}} \right|} \right|^2}\nonumber\\ &=& {\left( {{{\mathbf{\hat H}}^*}{\mathbf{b}_k} - {\mathbf{a}_k}} \right)^*}\left( {{{\mathbf{\hat H}}^*}{\mathbf{b}_k} - {\mathbf{a}_k}} \right) + \left( {\frac{{\sigma _h^2}}{2}\mathop \sum \limits_{l = 1}^K \left( {\text{Tr}\left( {\mathbf{V}_l^*{\mathbf{V}_l}} \right) + \text{Tr}\left( {\mathbf{V}_{\bar l}^*{\mathbf{V}_{\bar l}}} \right)} \right) + \frac{{\sigma _r^2}}{2}} \right)\mathbf{b}_k^*{\mathbf{b}_k}\nonumber\\ &=& \mathbf{b}_k^*\mathbf{\hat H}{\mathbf{\hat H}^*}{\mathbf{b}_k} - 2\mathbf{b}_k^*\mathbf{\hat H}{\mathbf{a}_k} + \mathbf{a}_k^*{\mathbf{a}_k} + \left( {\frac{{\sigma _h^2}}{2}\mathop \sum \limits_{l = 1}^K \left( {\text{Tr}\left( {\mathbf{V}_l^*{\mathbf{V}_l}} \right) + \text{Tr}\left( {\mathbf{V}_{\bar l}^*{\mathbf{V}_{\bar l}}} \right)} \right) + \frac{{\sigma _r^2}}{2}} \right)\mathbf{b}_k^*{\mathbf{b}_k}\nonumber\\ &=& \mathbf{b}_k^*\left( {\mathbf{\hat H}{{\mathbf{\hat H}}^*} + \frac{{\sigma _h^2}}{2}\mathop \sum \limits_{l = 1}^K \left( {\text{Tr}\left( {\mathbf{V}_l^*{\mathbf{V}_l}} \right) + \text{Tr}\left( {\mathbf{V}_{\bar l}^*{\mathbf{V}_{\bar l}}} \right)} \right)\mathbf{I} + \frac{{\sigma _r^2}}{2}\mathbf{I}} \right){\mathbf{b}_k} - 2\mathbf{b}_k^*\mathbf{\hat H}{\mathbf{a}_k} + \mathbf{a}_k^*{\mathbf{a}_k}.
\end{eqnarray}
The optimum value of $\mathbf{b}_k$ is the solution of
\begin{eqnarray}
\frac{{\partial f\left( {{\mathbf{b}_k},{\mathbf{a}_k}} \right)}}{{\partial {\mathbf{b}_k}}} = 2\left( {\mathbf{\hat H}{{\mathbf{\hat H}}^*} + \frac{{\sigma _h^2}}{2}\mathop \sum \limits_{l = 1}^K \left( {\text{Tr}\left( {\mathbf{V}_l^*{\mathbf{V}_l}} \right) + \text{Tr}\left( {\mathbf{V}_{\bar l}^*{\mathbf{V}_{\bar l}}} \right)} \right)\mathbf{I} + \frac{{\sigma _r^2}}{2}\mathbf{I}} \right){\mathbf{b}_k} - \mathbf{\hat H}{\mathbf{a}_k} = 0.
\end{eqnarray}
Hence,
\begin{eqnarray}
\mathbf{b}_k^* = \mathbf{a}_k^*{\mathbf{\hat H}^*}{\left( {\frac{{\sigma _r^2}}{2}\mathbf{I} + \frac{{\sigma _h^2}}{2}\left( {\mathop \sum \limits_{l = 1}^K \text{Tr}\left( {\mathbf{V}_l^*{\mathbf{V}_l}} \right) + \text{Tr}\left( {\mathbf{V}_{\bar l}^*{\mathbf{V}_{\bar l}}} \right)} \right)\mathbf{I} + \mathbf{\hat H}{{\mathbf{\hat H}}^*}} \right)^{ - 1}}\mathbf{\hat H}.
\end{eqnarray}
Thus, the theorem is proved.

\end{document}